\documentclass[onecolumn, 12pt]{IEEEtran}
\usepackage[dvipsnames]{xcolor}
\usepackage{color, graphics, url}
\usepackage{cite}
\usepackage{amsmath, amsfonts, amssymb, amsthm}
\usepackage{epsfig}
\usepackage{subfigure}
\usepackage{comment}
\usepackage{array}

\newcommand{\JHnewpage}{} %\newpage
\newcommand{\NN}{\nonumber}
\newcommand{\NNL}{\nonumber\\}

\newcommand{\SNR}{{\scriptscriptstyle\textsf{snr}}}
\newcommand{\INR}{{\scriptscriptstyle\textsf{inr}}}
\newcommand{\OIA}{{\scriptscriptstyle\textsf{OIA}}}
\newcommand{\SINR}{{\scriptscriptstyle\textsf{sinr}}}

\newcommand{\DoF}[1]{\lim_{P\to\infty}\frac{#1}{\log_2 P}}
\newcommand{\tDoF}[1]{\underset{P\to\infty}{\lim} \frac{#1}{\log_2 P}}
\newcommand{\h}[1]{\mathbf{h}_{#1}}

\newcommand{\PR}[1]{\mathrm{Pr}\left[#1\right]}
\newcommand{\argmin}[1]{\underset{#1}{\arg\min}}
\newcommand{\argmax}[1]{\underset{#1}{\arg\max}}

\newcommand{\R}{\mathcal{R}}

\newtheorem{theorem}{Theorem}
\newtheorem{lemma}{Lemma}

\newtheorem{remark}{Remark}

\title{Multiuser Diversity in Interfering Broadcast Channels: Achievable
Degrees of Freedom and User Scaling Law}

\author{\IEEEauthorblockN{Jung~Hoon~Lee},~\IEEEmembership{Student
Member,~IEEE}, \IEEEauthorblockN{Wan~Choi},~\IEEEmembership{Senior
Member,~IEEE}, and
\IEEEauthorblockN{Bhaskar~D.~Rao},~\IEEEmembership{Fellow,~IEEE}
\thanks{This work was supported by the Korea Research Foundation Grant
funded by the Korean Government (NRF-2012-013-2012S1A2A1A01031507).}
\thanks{J. H. Lee and W.~Choi are with Department of Electrical Engineering,
        Korea Advanced Institute of Science and Technology (KAIST), Daejeon
        305-701, Korea (e-mail: tantheta@kaist.ac.kr, wchoi@ee.kaist.ac.kr)}
\thanks{B. D. Rao is with Department of Electrical and Computer
Engineering, University of California, San Diego, La Jolla, CA
92093-0407 USA (e-mail: brao@ece.ucsd.edu).}}

\begin{document}
\maketitle

\vspace{-.6in}
\begin{abstract}
This paper investigates how multiuser dimensions can effectively be
exploited for target degrees of freedom (DoF) in interfering
broadcast channels (IBC) consisting of $K$-transmitters and their
user groups. First, each transmitter is assumed to have a single
antenna and serve a singe user in its user group where each user has
receive antennas less than $K$. In this case, a $K$-transmitter
single-input multiple-output (SIMO) interference channel (IC) is
constituted after user selection. Without help of multiuser
diversity, $K-1$ interfering signals cannot be perfectly removed at
each user since the number of receive antennas is smaller than or
equal to the number of interferers. Only with proper user selection,
non-zero DoF per transmitter is achievable as the number of users
increases. Through geometric interpretation of interfering channels,
we show that the multiuser dimensions have to be used first for
reducing the DoF loss caused by the interfering signals, and then
have to be used for increasing the DoF gain from its own signal. The
sufficient number of users for the target DoF is derived. We also
discuss how the optimal strategy of exploiting multiuser diversity
can be realized by practical user selection schemes. Finally, the
single transmit antenna case is extended to the multiple-input
multiple-output (MIMO) IBC where each transmitter with multiple
antennas serves multiple users.
\end{abstract}

\begin{IEEEkeywords}
Multiuser diversity, degrees of freedom, interference alignment
measure, interfering broadcast channel
\end{IEEEkeywords}

\newpage

%%%%%%%%%%%%%%%%%%%%%%%%%%%%%%%%%%%%%%%%%%%%%%%%%%%%%%%%%%%%%%%%%%%
% % % % % % % % % % % % % % % % % % % % % % % % % % % % % % % % % %
%%%%%%%%%%%%%%%%%%%%%%%%%%%%%%%%%%%%%%%%%%%%%%%%%%%%%%%%%%%%%%%%%%%
%%%%%%%%%%%%%%%%%%%%%%%%%%%%%%%%%%%%%%%%%%%%%%%%%%%%%%%%%%%%%%%%%%%
% % % % % % % % % % % % % % % % % % % % % % % % % % % % % % % % % %
%%%%%%%%%%%%%%%%%%%%%%%%%%%%%%%%%%%%%%%%%%%%%%%%%%%%%%%%%%%%%%%%%%%
\section{Introduction}
Interference is a major performance-limiting factor in modern
wireless communication systems. Many interference mitigation
strategies have been proposed to improve network spectral
efficiency. By allowing partial or full cooperation among
interfering base stations, interference can effectively be managed
and spectral efficiency can be improved. Joint beamforming
\cite{DY2010} and network MIMO (or multicell processing)
\cite{GHHSSY2010} among base stations have been shown to be
effective interference mitigation techniques. However, if
cooperation among transmitters is not allowed, orthogonal multiple
access has been a traditional solution to interference. In a
$K$-user single-input single-output (SISO) interference channel
(IC), for example, each user can achieve $1/K$ degrees of freedom
(DoF) by time division multiple access.

In recent years, interference alignment (IA) techniques have
received much attention \cite{CJ2008, GCJ2008, ST2008, GJ2008}. The
basic concept of IA is to align the interfering signals in a small
dimensional subspace. In a $K$-user SISO IC, $K/2$ DoF have been
shown to be achievable using IA \cite{CJ2008}. Although IA provides
a substantial asymptotic capacity gain in interference channels,
there are many practical challenges for implementation.
IA requires global channel state information at the transmitter
(CSIT), and imperfect channel knowledge severely degrades the gain
of IA. In some channel configurations, symbols should be extended in
the time/frequency domain to align interfering signals. The high
computational complexity is also considered as a major challenge. To
ameliorate these difficulties, many IA algorithms have been proposed
such as iterative IA \cite{GCJ2008} and a subspace IA \cite{ST2008}.

For interference suppression, multiuser diversity can also be
exploited by opportunistic user selection for minimizing
interference. The interference reduction by multiuser diversity can
be enjoyed without heavy burden on global channel knowledge because
user selection in general requires only a small amount of feedback
\cite{SH2005, HR2013, JPS2012, LC2010, LC2011, LC2013, LC2011_2,
KLC2011}.
In this context, opportunistic interference alignment (OIA) has been
recently proposed in \cite{LC2010} and has attracted much attention.
In a 3-transmitter $M\times 2M$ MIMO interfering broadcast channel
(IBC), the authors of \cite{LC2013} proved that $\alpha M$ (where
$\alpha\in[0,1]$) DoF per transmitter is achievable when the number
of users scales as $P^{\alpha M}$. In \cite{LC2011_2} and
\cite{KLC2011}, a $K$-transmitter $1\times 3$ SIMO IBC and a
$K$-transmitter $1\times (K-1)$ SIMO IBC have been studied,
respectively. For SIMO interfering multiple access channel (IMAC)
constituted by $K$-cell uplink channels with $M$ transmit antennas
and single antenna users, the authors of \cite{JPS2012} showed that
$KM$ DoF are achievable when the number of users scales as
$P^{(K-1)M}$.
In these schemes, user dimensions are used to align the interfering
signals; each transmitter opportunistically selects a user whose
interfering signals are most aligned among the users associated with
the transmitter. Contrary to the conventional opportunistic user
selection techniques \cite{VTL2002, SH2005, HR2013, CFAH2007,
CA2008, JS2010}, the OIA scheme exploits the multiuser dimensions
for interference alignment.

%%%%%%%%%%%%%%%%%%%%%%%%%%%%%%%%%%%%%%%%%%%%%%%%%%%%%%%%%%%%%%%%%%%
% % % % % % % % % % % % % % % % % % % % % % % % % % % % % % % % % %
%%%%%%%%%%%%%%%%%%%%%%%%%%%%%%%%%%%%%%%%%%%%%%%%%%%%%%%%%%%%%%%%%%%
In this paper, we investigate the optimal role of multiuser
diversity for the target DoF in the IBC with $K$-transmitters and
generalize the results of \cite{LC2011_2}\cite{KLC2011}.
For the $K$-transmitter SIMO IBC, each transmitter selects and
serves a single user in its user group consisting of $N$ users. Once
after $K$-transmitters select their serving users, a $K$-user IC is
constructed. Each user has $N_r$ antennas less than or equal to the
number of interferers, i.e., $N_r \le K-1$. Thus, without help of
multiuser diversity, interference at each user cannot be perfectly
removed so that the achievable rate of each transmitter goes to zero
as signal-to-noise ratio (SNR) increases. Consequently, the
achievable DoF per transmitter becomes zero. However, non-zero DoF
per transmitter is achievable by exploiting multiuser diversity as
the number of users increases.

Since opportunistic user selection can focus on either enhancing the
desired signal or decreasing interference, non-zero DoF can be
obtained by properly enhancing the desired signal strength and
reducing interference via user selection. That is, the non-zero DoF
$d$ comprises a DoF gain term $d_1\ge0$ from the desired signal and
a DoF loss term $d_2\ge0$ caused by interference such that $d_1 -
d_2 = d$, and the target DoF $d$ can be obtained by a proper
combination of $d_1$ and $d_2$. However, many questions remain
unsolved; what is the feasible and optimal combination of $(d_1,
d_2)$ for the target DoF $d ~(=d_1-d_2)$ and what is the sufficient
number of users for the target DoF achieving strategy. We answer
these fundamental questions and analytically investigate how the
multiuser dimensions can be optimally exploited for the target DoF
in the IBC. Specifically, from geometric interpretation of
interfering channels, we define an interference alignment measure
that indicates how well interference signals are aligned at each
user.

Using the interference alignment measure, we first consider the $K$-
transmitters SIMO IBC and show that the DoF gain term $d_1$ can be
achieved if the number of users scales in terms of transmit power
$P$ as $N\propto e^{P^{(d_1-1)}}$ and the DoF loss term can be
reduced to $d_2$ if the number of users scales as $N\propto
P^{(1-d_2)(K-N_r)}$. From these results, we find the optimal
strategy of exploiting multiuser diversity for the target DoF $d$ in
terms of the required number of users; the optimal target DoF
achieving strategies $(d_1^\star, d_2^\star)$ are $(1,1-d)$ and
$(d,0)$ for the target DoF $d\in[0,1]$ and $d ~(> 1)$, respectively.
We also investigate how the optimal target DoF achieving strategy
$(d_1^\star, d_2^\star)$ can be realized by practical user selection
schemes. Then, we extend our results to the $K$-transmitter MIMO IBC
where each transmitter has $N_t$ multiple antennas and serves
multiple users with $N_r$ receive antenna each. Our generalized key
findings are summarized as follows:
\begin{itemize}
\item For the target DoF $d\in[0,N_t]$, $(d_1^\star, d_2^\star) = (N_t,N_t-d)
$ is the optimal target DoF achieving strategy that minimizes the
required number of users. That is, the multiuser dimensions should
be exploited to make the DoF loss $N_t-d$. The sufficient number of
users for this strategy scales like $N\propto P^{(d/N_t)
(KN_t-N_r)}$.
\item For the target DoF $d ~(>N_t)$, $(d_1^\star, d_2^\star) = (d,0) $ is the
optimal target DoF achieving strategy which minimizes the required
number of users. That is, the multiuser dimensions should be
exploited to make the DoF loss term zero as well as to make the DoF
gain term $d$. The sufficient number of users for this strategy
scales like $N\propto e^{P^{(d/N_t-1)}}P^{(KN_t-N_r)}$.
\end{itemize}
The rest of this paper is organized as follows. In Section II, we
describe the system model. In Section III, a geometric
interpretation of interfering channels is provided, and the
interference alignment measure is defined. Section IV derives the
optimal strategies of achieving the target DoF in terms of the
required number of users. In Section V, we show how various
practical user selection schemes exploit multiuser diversity for the
target DoF and discuss their optimality to achieve the target DoF.
The system model is extended for the MIMO IBC in Section VI.
Numerical results are shown in Section VII, and we conclude our
paper in Section VIII.

%%%%%%%%%%%%%%%%%%%%%%%%%%%%%%%%%%%%%%%%%%%%%%%%%%%%%%%%%%%%%%%%%%%
% % % % % % % % % % % % % % % % % % % % % % % % % % % % % % % % % %
%%%%%%%%%%%%%%%%%%%%%%%%%%%%%%%%%%%%%%%%%%%%%%%%%%%%%%%%%%%%%%%%%%%
\noindent -- \emph{Notations}

Throughout the paper, we use boldface to denote vectors and
matrices. The notations $\mathbf{A}^\dagger$,
$\Lambda_i(\mathbf{A})$, and $V_i(\mathbf{A})$ denote the conjugate
transpose, the $i$th largest eigenvalue, and the eigenvector of
matrix $\mathbf{A}$ corresponding to the $i$th largest eigenvalue.
For convenience, the smallest eigenvalue, the largest eigenvalue,
and the eigenvectors corresponding eigenvectors of $\mathbf{A}$ are
denoted as $\Lambda_{\min}(\mathbf{A})$, $\Lambda_{\max}
(\mathbf{A})$, $V_{\min}(\mathbf{A})$, and $V_{\max}(\mathbf{A})$,
respectively.
Also, $\mathbf{I}_{n}$, $\mathbb{C}^{n}$, and $\mathbb{C}^{m\times
n}$ indicate the $n\times n$ identity matrix, the $n$-dimensional
complex space, and the set of $m\times n$ complex matrices,
respectively.

%%%%%%%%%%%%%%%%%%%%%%%%%%%%%%%%%%%%%%%%%%%%%%%%%%%%%%%%%%%%%%%%%%%
% % % % % % % % % % % % % % % % % % % % % % % % % % % % % % % % % %
%%%%%%%%%%%%%%%%%%%%%%%%%%%%%%%%%%%%%%%%%%%%%%%%%%%%%%%%%%%%%%%%%%%
\JHnewpage
\section{Problem Formulation}
%%%%%%%%%%%%%%%%%%%%%%%%%%%%%%%%%%%%%%%%%%%%%%%%%%%%%%%%%%%%%%%%%%%
% % % % % % % % % % % % % % % % % % % % % % % % % % % % % % % % % %
%%%%%%%%%%%%%%%%%%%%%%%%%%%%%%%%%%%%%%%%%%%%%%%%%%%%%%%%%%%%%%%%%%%
\subsection{System Model}\label{sec:system_model}

%%%\begin{figure}[!t]
%%%\centering
%%%  \includegraphics[width=.777\columnwidth]{./figures/OIA_K_SIMO.eps}\\
%%%  \caption{System model. Each transmitter selects and serves a single user in each group.}
%%%  \label{fig:system_model}
%%%\end{figure}

Our system model is depicted in Fig. \ref{fig:system_model}. The
system corresponds to the interfering broadcast channel (IBC) of
which capacity is unknown.
There are $K$ transmitters with $N_t$ transmit antennas each, and
each transmitter has its own user group consisting of $N$ users with
$N_r$ antennas each.
First, each transmitter is assumed to have a single antenna, i.e.,
$N_t=1$, and serves a single user selected in its user group so that
$K$-transmitter SIMO IC is opportunistically constituted.
The system model with multiple transmit antennas (i.e., $N_t>1$)
becomes statistically identical with the single transmit antenna
model if each transmitter uses a random precoding vector.
In Section \ref{sec:extension}, we extend our system model to the
$K$-transmitter MIMO IBC where each transmitter with multiple
antennas serves the multiple users through orthonormal random beams.
%

%Note that we consider single user support per transmitter at each
%time rather than simultaneous multiuser support in order to
%effectively and tractably capture the effects of multiuser diversity
%in the IBC.

In this paper, we focus on the cases that the number of receive
antennas is smaller than the number of transmitters, i.e., $N_r <
K$. Otherwise (i.e., if $N_r \ge K$), each user can suppress all
interfering signals through zero-forcing like schemes so that DoF
one is trivially guaranteed at each transmitter. We also assume that
collaboration or information sharing among the transmitters is not
allowed. Since the user selection at each transmitter is independent
of the other transmitters', we only consider the achievable rate of
the first transmitter without loss of generality. Note that the
average achievable rate per transmitter will be same if the
configurations of the transmitters are identical.

At the first transmitter, the received signal at the $n$th user
denoted by $\mathbf{y}_n\in\mathbb{C}^{N_r\times1}$ is given by
\begin{align}
    \mathbf{y}_{n}
    &=\mathbf{h}_{n,1} x_1
        + \sum_{k=2}^{K}\mathbf{h}_{n,k} x_k
        + \mathbf{z}_n,\NN
\end{align}
where $\mathbf{h}_{n,k} \in \mathbb{C}^{N_r\times 1}$ is the vector
channel from the $k$th transmitter to the $n$th user whose elements
are independent and identically distributed (i.i.d.) circularly
symmetric complex Gaussian random variables with zero means and unit
variance. Also, $x_k\in \mathbb{C}^{1\times 1}$ is the transmitted
signal using random Gaussian codebook from the $k$th transmitter
such that $\mathbb{E}\vert x_k\vert^2 = P$, where $P$ is the power
budget at each transmitter. Also, $\mathbf{z}_n \in \mathbb{C} ^{N_r
\times 1}$ is a circularly symmetric complex Gaussian noise with
zero mean and an identity covariance matrix, i.e., $\mathbf{z}_n
\sim \mathcal{CN} (0,\mathbf{I}_{N_r})$. Assuming perfect channel
estimation at each receiver, the channel state information
$\{\mathbf{h} _{n,k}\}_{k=1}^K$ is available at the $n$th user.

The received signal is postprocessed at each user using multiple
receive antennas.
%
%Various schemes are possible; the receive antennas can be used to
%increase the own signal power, to decrease the interference power,
%or etc.
%
Let $\mathbf{v}_n \in \mathbb{C}^{N_r \times 1}$ be the
postprocessing vector of the $n$th user such that $\Vert
\mathbf{v}_n\Vert^2=1$.
Then, the received signal after postprocessing becomes
\begin{align}
 \mathbf{v}_n^\dagger\mathbf{y}_{n}
    &=\mathbf{v}_n^\dagger\mathbf{h}_{n,1} x_1
        + \sum_{k=2}^{K}\mathbf{v}_n^\dagger \mathbf{h}_{n,k} x_k
        + \mathbf{v}_n^\dagger\mathbf{z}_n.
        \label{eqn:vy_n}
\end{align}

To aid user selection at the transmitter, each user feeds one scalar
value back to the transmitter. Various user selection criteria and
corresponding feedback information will be discussed in the
following sections.
Since no information is shared among the transmitters, each
transmitter independently selects a single user based on the
collected information.

Let $n^*$ be the index of the selected user at the first
transmitter.
Then, the average achievable rate of the first transmitter is given
by
\begin{align}
 \R \triangleq \mathbb{E}\log_2
    \left(1+
    \frac{ P\vert \mathbf{v}_{n^*}^\dagger\h{n^*,1} \vert^2}
    {
        1 + P \sum_{k=2}^{K} \vert
        \mathbf{v}_{n^*}^\dagger \h{n^*,k} \vert^2
    }\right). \label{eqn:C_n}
\end{align}
We decompose $\R$ into two terms $\R^+$ and $\R^-$ such that $\R =
\R^+ - \R^-$, which are given, respectively, by
\begin{align}
 \R^+
    &= \mathbb{E}\log_2\left(1+  P \sum_{k=1}^{K}
    \vert \mathbf{v}_{n^*}^\dagger \h{n^*,k} \vert^2\right),
    \label{eqn:Rgain}\\
 \R^-
    &= \mathbb{E}\log_2\left(1+  P \sum_{k=2}^{K}
    \vert \mathbf{v}_{n^*}^\dagger \h{n^*,k} \vert^2\right).
    \label{eqn:Rloss}
\end{align}
Then, the achievable DoF of the first transmitter becomes
\begin{align}
 \DoF{\R} =  \DoF{\R^+} - \DoF{\R^-}.
\end{align}
We call $\tDoF{\R^+}$ and $\tDoF{\R^-}$ as \emph{DoF gain term} and
\emph{DoF loss term}, respectively.

%%%%%%%%%%%%%%%%%%%%%%%%%%%%%%%%%%%%%%%%%%%%%%%%%%%%%%%%%%%%%%%%%%%
% % % % % % % % % % % % % % % % % % % % % % % % % % % % % % % % % %
%%%%%%%%%%%%%%%%%%%%%%%%%%%%%%%%%%%%%%%%%%%%%%%%%%%%%%%%%%%%%%%%%%%
\subsection{Problem Description}\label{sec:problem_description}

The achievable rate of each transmitter depends on the number of
users because multiuser dimensions are exploited for a rate
increase. When there are fixed number of users, the achievable rate
of each transmitter will be saturated in the high SNR region due to
interferences because the number of receive antennas at each user is
smaller than the number of total transmitters. Consequently, the
first transmitter cannot obtain any DoF, i.e.,
\begin{align}
 \lim_{P\to\infty \atop \textrm{Fixed~}N}
    \frac{\R}{\log_2 P} = 0.
    \label{eqn:zero_DoF}
\end{align}
In this case, both the DoF gain term and the DoF loss term become
one, i.e.,
\begin{align}
 \left(\lim_{P\to\infty \atop \textrm{Fixed~}N} \frac{\R^+}{\log_2P},
 \lim_{P\to\infty \atop \textrm{Fixed~}N} \frac{\R^-}{\log_2P}\right)
 = (1,1).
 \label{eqn:DoF_fixedN}
\end{align}
On the other hand, when the transmit power is fixed, the achievable
rate of the selected user can increase to infinity as the number of
users increases, i.e.,
\begin{align}
 \lim_{N\to \infty \atop \textrm{Fixed~}P}\R = \infty.
    \label{eqn:infinity_rate}
\end{align}
Then, how much DoF can be achieved when both the number of users and
the transmit power increase?
Obviously, non-zero DoF can be obtained by exploiting multiuser
dimensions, and the achievable DoF
\begin{align}
 \lim_{N\to\infty}
 \left[\lim_{P\to\infty}
     \frac{\R}{\log_2P}
 \right]
 \label{eqn:opportunistic_DoF}
\end{align}%
will depend on the increasing speeds of $N$ and $P$.
In this case, DoF $d~(>0)$ at the first transmitter comprises the
DoF gain term $d_1~(\ge0)$ and the DoF loss term $d_2~(\ge0)$ such
that $d_1-d_2=d$, i.e.,
\begin{align}
 (d_1, d_2) \triangleq \left(\DoF{\R^+}, \DoF{\R^-}\right).
\end{align}
We call $(d_1, d_2)$ as a \emph{target DoF achieving strategy} if $d
= d_1-d_2$ for the target DoF $d$.
Since each strategy requires different user scaling, we need to find
the optimal DoF achieving strategy that exploits multiuser diversity
most efficiently, i.e., which requires the minimum user scaling. For
the \emph{target DoF} per transmitter $d~ (>0)$, we find the optimal
target DoF achieving strategy $(d_1^\star, d_2^\star)$ satisfying
$d_1^\star - d_2^\star=d$ and derive the required user scaling.
Note that the definition of DoF in this paper is extended from the
conventional definition of DoF in order to properly capture
multiuser diversity gain in terms of achievable rate. Achievable DoF
defined in \eqref{eqn:opportunistic_DoF} depends on increasing
speeds of $N$ and $P$; and can have non-zero values even larger than
one if the number of users properly scales with the transmit power.

%\newpage
%%%%%%%%%%%%%%%%%%%%%%%%%%%%%%%%%%%%%%%%%%%%%%%%%%%%%%%%%%%%%%%%%%%
% % % % % % % % % % % % % % % % % % % % % % % % % % % % % % % % % %
%%%%%%%%%%%%%%%%%%%%%%%%%%%%%%%%%%%%%%%%%%%%%%%%%%%%%%%%%%%%%%%%%%%
\subsection{DoF Achieving Strategies and Reduced Set of Candidates for the Optimal Strategy}

From the definitions of the rate gain term and the rate loss term
given in \eqref{eqn:Rgain} and \eqref{eqn:Rloss}, respectively, the
strategies which achieve the target DoF $d$ are given by
\begin{align}
 \big\{
    (d_1, d_2) ~\vert~ d_1-d_2=d,~d_1\ge d_2,~d_1\ge0,~d_2\ge 0
 \big\}.\label{eqn:DoF_strategy_set0}
\end{align}

The following lemma shows that we do not need to consider all of the
candidate strategies in \eqref{eqn:DoF_strategy_set0} but take into
account only a subset of \eqref{eqn:DoF_strategy_set0} to find the
optimal target DoF achieving strategy.

%%%%%%%%%%%%%%%%%%%%%%%%%%%%%%%%%%%%%%%%%%%%%%%%%%%%%%%%%%%%%%%%%%%
% % % % % % % % % % % % % % % % % % % % % % % % % % % % % % % % % %
%%%%%%%%%%%%%%%%%%%%%%%%%%%%%%%%%%%%%%%%%%%%%%%%%%%%%%%%%%%%%%%%%%%
\begin{lemma}\label{lemma:DoF_terms_range}
For any non-negative target DoF, the optimal DoF achieving strategy
is in the set
\begin{align}
 \big\{
    (d_1, d_2) ~\vert~ d_1\in[1,\infty),~d_2\in[0,1],~d_1-d_2=d
 \big\}.
 \label{eqn:DoF_term_range}
\end{align}
\end{lemma}
%%%%%%%%%%%%%%%%%%%%%%%%%%%%%%%%%%%%%%%%%%%%%%%%%%%%%%%%%%%%%%%%%%%
\begin{proof}
At each channel realization, the achievable DoF has the form of
\begin{align}
 \log_2\left(1+X+Y\right) - \log_2\left(1+Y\right),
 \label{eqn:function0}
\end{align}
where $X\triangleq \vert \mathbf{v}_{n^*}^\dagger \h{n^*,1} \vert^2$
and $Y\triangleq P \sum_{k=2}^{K} \vert \mathbf{v}_{n^*} ^\dagger
\h{n^*,k} \vert^2$ are its own signal power and the interfering
signal power at the selected user, respectively.
Since the function \eqref{eqn:function0} is an increasing function
of $X$ and a decreasing function of $Y$, for an increase of
\eqref{eqn:function0}, the multiuser dimension should be used for
increasing $X$, for decreasing $Y$, or mixture of them.
This fact results in \eqref{eqn:DoF_term_range}.
\end{proof}

Lemma 1 provides a basic guideline of using the multiuser dimension;
multiuser diversity should not be used for either decreasing DoF
gain term or increasing DoF loss term. Since the optimal target DoF
achieving strategy is obtained in the reduced set of candidate
strategies, we consider the DoF gain term larger than one and DoF
loss term smaller than one, i.e., $d_1\in[1,\infty)$ and
$d_2\in[0,1]$, in the latter parts of this paper.

%%%%%%%%%%%%%%%%%%%%%%%%%%%%%%%%%%%%%%%%%%%%%%%%%%%%%%%%%%%%%%%%%%%
% % % % % % % % % % % % % % % % % % % % % % % % % % % % % % % % % %
%%%%%%%%%%%%%%%%%%%%%%%%%%%%%%%%%%%%%%%%%%%%%%%%%%%%%%%%%%%%%%%%%%%

%%%%%%%%%%%%%%%%%%%%%%%%%%%%%%%%%%%%%%%%%%%%%%%%%%%%%%%%%%%%%%%%%%%
% % % % % % % % % % % % % % % % % % % % % % % % % % % % % % % % % %
%%%%%%%%%%%%%%%%%%%%%%%%%%%%%%%%%%%%%%%%%%%%%%%%%%%%%%%%%%%%%%%%%%%
\JHnewpage
\section{Interference Alignment Measure}

%%%%%%%%%%%%%%%%%%%%%%%%%%%%%%%%%%%%%%%%%%%%%%%%%%%%%%%%%%%%%%%%%%%
% % % % % % % % % % % % % % % % % % % % % % % % % % % % % % % % % %
%%%%%%%%%%%%%%%%%%%%%%%%%%%%%%%%%%%%%%%%%%%%%%%%%%%%%%%%%%%%%%%%%%%
\subsection{Where does the DoF Loss Come from?}

In our system model, each user suffers from $K-1$ interfering
channels which is larger than or equal to the number of receive
antennas, i.e., $K-1\ge N_r$. Since the interfering channels are
isotropic and independent of each other, they span $N_r$-dimensional
space. Thus, the whole signal space at the receiver is corrupted by
interfering signals, and hence the DoF loss term becomes one if no
effort is made to align interfering signals. On the other hand, the
DoF loss can be reduced by aligning interfering signals in smaller
dimensional subspace.
%
%Then, the residual dimension can be used for desired signal.
%
For example, if the interfering signals are perfectly aligned in
$(N_r-1)$-dimensional subspace, they can be nullified by
postprocessing so that we can make the DoF loss zero.

The transmitter can exploit the multiuser dimensions to align
interfering signals by simply selecting a user whose interfering
channels are most aligned. Thus, each user needs to measure how much
the interfering channels are aligned in $(N_r-1)$-dimensional
subspace at the receiver. We call this measure as the
\emph{interference alignment measure}.
In this section, we geometrically interpret the interfering channels
and define the interference alignment measure at each user. The
interference alignment measure will be used for computing the
reducible DoF loss via multiuser diversity in Section IV.

%%%%%%%%%%%%%%%%%%%%%%%%%%%%%%%%%%%%%%%%%%%%%%%%%%%%%%%%%%%%%%%%%%%
% % % % % % % % % % % % % % % % % % % % % % % % % % % % % % % % % %
%%%%%%%%%%%%%%%%%%%%%%%%%%%%%%%%%%%%%%%%%%%%%%%%%%%%%%%%%%%%%%%%%%%
\subsection{Preliminaries}

%%%\begin{figure}[!t]
%%%\centering
%%%  \includegraphics[width=.977\columnwidth]{./figures/sphere_S1_S2.eps}
%%%  \caption{Illustration of $S_1(\mathbf{c}, \lambda)$ and $S_2(\mathbf{c}, \lambda)$
%%%  in $\mathbb{R}^3$ case.}
%%%  \label{fig:sphere_S1S2}
%%%\end{figure}

Let $S_0$ be the surface of the $N_r$-dimensional unit hypersphere
centered at the origin, i.e.,
\begin{align}
    S_0=\{\mathbf{x}\in\mathbb{C}^{N_r} \vert~ \Vert \mathbf{x}\Vert^2=1 \}.\NN
\end{align}
For an arbitrary unit vector $\mathbf{c}\in\mathbb{C}^{N_r}$ and an
arbitrary non-negative real number $0\le \lambda\le 1$, we can
divide $S_0$ into two parts, $S_1 (\mathbf{c}, \lambda)$ and $S_2
(\mathbf{c}, \lambda)$, given by
\begin{align}
 S_1 (\mathbf{c}, \lambda)
 &\triangleq \big\{\mathbf{x} \in\mathbb{C}^{N_r} \big\vert~
        \vert \mathbf{c}^\dagger\mathbf{x} \vert^2 \ge \lambda
        ,~~\Vert \mathbf{x}\Vert^2=1 \big\}\NNL
 S_2 (\mathbf{c}, \lambda)
 &\triangleq \big\{\mathbf{x} \in\mathbb{C}^{N_r} \big\vert~
        \vert \mathbf{c}^\dagger\mathbf{x} \vert^2 \le \lambda,
        ~~\Vert \mathbf{x}\Vert^2=1
        \big\}.\label{eqn:S2}
\end{align}
When $\mathbf{x}, \mathbf{c} \in \mathbb{R}^3$, two parts
$S_1(\mathbf{c}, \lambda)$ and $S_2(\mathbf{c}, \lambda)$ are
represented in Fig. \ref{fig:sphere_S1S2}.
Let $A\left(S_i (\mathbf{c}, \lambda)\right)$ be the surface area of
$S_i (\mathbf{c}, \lambda)$ for $i=0,1,2$.
The surface area of an $N_r$-dimensional complex unit hypersphere is
given by $A(S_0)=2\pi^{N_r}/(N_r-1)!$, and it was shown that
\cite{MSEA2003}
\begin{align}
    A\left(S_1(\mathbf{c}, \lambda)\right)
    =\frac{2\pi^{N_r}(1-\lambda)^{N_r-1}}{(N_r-1)!}, \NN
\end{align}
which is invariant with $\mathbf{c}$.
Therefore, we obtain
\begin{align}
 A\left(S_2(\mathbf{c}, \lambda)\right)
    =\frac{2\pi^{N_r}(1-(1-\lambda)^{N_r-1})}{(N_r-1)!}\NN
\end{align}
from the relationship $A(S_0) = A\left(S_1 (\mathbf{c},
\lambda)\right) + A\left(S_2 (\mathbf{c}, \lambda)\right)$.
From this fact, we obtain the following lemma.

%%%%%%%%%%%%%%%%%%%%%%%%%%%%%%%%%%%%%%%%%%%%%%%%%%%%%%%%%%%%%%%%%
\begin{lemma} \label{lemma:p_n_lambda}
Let $\mathbf{g}_1, \ldots, \mathbf{g}_m$ be independent and
isotropic unit vectors in $\mathbb{C}^{N_r}$.
For an arbitrary unit vector $\mathbf{c} \in \mathbb{C}^{N_r}$ and
$\lambda \in [0,1]$, the probability that $S_2( \mathbf{c},
\lambda)$ contains $\{\mathbf{g}_1, \ldots, \mathbf{g}_m\}$ becomes
\begin{align}
 \Pr[\{\mathbf{g}_1, \ldots, \mathbf{g}_m\} \subset
    S_2( \mathbf{c}, \lambda) ]
   = \left(1-(1-\lambda)^{N_r-1}\right)^m,
    \label{eqn:prob_m_lambda}
\end{align}
which is invariant with $\mathbf{c}$.
\end{lemma}\vspace{.1in}
%%%%%%%%%%%%%%%%%%%%%%%%%%%%%%%%%%%%%%%%%%%%%%%%%%%%%%%%%%%%%%%%%
\begin{proof}
From the ratio of $A(S_2(\mathbf{c}, \lambda))$ and $A(S_0)$, we
obtain
\begin{align}
 \Pr[\mathbf{g}_i \in S_2( \mathbf{c}, \lambda)]
   = \frac{A(S_2(\mathbf{c}, \lambda))}{A(S_0)}
   =1-(1-\lambda)^{N_r-1},
   \quad \forall i.
\end{align}
Since $\mathbf{g}_1, \ldots, \mathbf{g}_m$ are independent of each
other, it is satisfied that
\begin{align}
 \Pr[\{\mathbf{g}_1, \ldots, \mathbf{g}_m\}
    \subset S_2( \mathbf{c}, \lambda) ]
 &= \Pr[\mathbf{g}_1 \in S_2( \mathbf{c}, \lambda)]^m, \NN
\end{align}
which is given in \eqref{eqn:prob_m_lambda}.
\end{proof}
%%%%%%%%%%%%%%%%%%%%%%%%%%%%%%%%%%%%%%%%%%%%%%%%%%%%%%%%%%%%%%%%%
\vspace{.1in}

%%%%%%%%%%%%%%%%%%%%%%%%%%%%%%%%%%%%%%%%%%%%%%%%%%%%%%%%%%%%%%%%%%%
% % % % % % % % % % % % % % % % % % % % % % % % % % % % % % % % % %
%%%%%%%%%%%%%%%%%%%%%%%%%%%%%%%%%%%%%%%%%%%%%%%%%%%%%%%%%%%%%%%%%%%
\subsection{Interference Alignment Measure at Each User}

In this subsection, we define the interference alignment measure at
each user. The DoF loss is determined by how much the interfering
channels are closely aligned in $(N_r-1)$-dimensional subspace. Only
if interfering channels are perfectly aligned in
$(N_r-1)$-dimensional subspace, we can have zero DoF loss. The
interference alignment measure is used for computing the DoF loss at
each user.
Let $\mathbf{g}_1, \ldots, \mathbf{g}_{K-1}$ be the $K-1~(\ge N_r)$
normalized interfering channels at a user and
$\mathfrak{q}(\mathbf{g}_1, \ldots, \mathbf{g}_{K-1})$ be the
interference alignment measure among them.

Consider the following optimization problem:
\begin{align}
\underset{\mathbf{c},\lambda}{\textrm{minimize}}
    &\qquad A(S_2(\mathbf{c}, \lambda))
    \label{eqn:c_opt_lambda_opt}\\
\textrm{subject to}
    &\qquad
    S_2(\mathbf{c},\lambda)\supset \{\mathbf{g}_1, \ldots, \mathbf{g}_{K-1}\},\NNL
    &\qquad \Vert \mathbf{c} \Vert^2=1, \quad  \lambda \in [0,1].\NN
\end{align}
From the definition of $S_2(\mathbf{c},\lambda)$ given in
\eqref{eqn:S2}, this problem is equivalent to
\begin{align}
\textrm{minimize}&\qquad \lambda \label{eqn:optimization_prob}\\
\textrm{subject to}
    & \qquad \vert \mathbf{c} ^\dagger \mathbf{g}_k \vert^2 \le
    \lambda \quad  \textrm{~for~~}1\le k\le K-1, \NNL
    & \qquad \Vert\mathbf{c} \Vert^2=1,~~     \lambda \in[0, 1],\NN
\end{align}
which can be solved by linear programming \cite{LP1, LP2}.
Let $(\mathbf{c}^\star , \lambda^\star )$ be the solution of the
above problem. Then, $S_2( \mathbf{c} ^\star , \lambda^\star )$ has
the smallest surface area among all $S_2(\mathbf{c}, \lambda )$
containing $\mathbf{g}_1, \ldots, \mathbf{g}_{K-1}$.

%%%\begin{figure}[!t]
%%%\centering
%%%  \includegraphics[width=\columnwidth]{./figures/IAM_concept.eps}
%%%  \caption{Graphical representations of the interfering channels
%%%  and interference alignment measures.}
%%%  \label{fig:IAM_concept}
%%%\end{figure}

Using $\mathbf{c}^\star $, we can divide an $N_r$-dimensional space
into two subspaces which are the one-dimensional subspace spanned by
$\mathbf{c}^\star $ and the $(N_r-1)$-dimensional complementary
subspace denoted by $\mathcal{U}$.
If there exists $\mathbf{c}^\star $ such that $\mathbf{c}^\star
\perp \{ \mathbf{g}_1, \ldots, \mathbf{g}_{K-1} \}$, it is satisfied
that $\textrm{span}( \mathbf{g}_1, \ldots, \mathbf{g}_{K-1} )\subset
\mathcal{U}$ and $S_2( \mathbf{c} ^\star , 0) \supset \{
\mathbf{g}_1, \ldots, \mathbf{g}_{K-1} \}$, and hence
$\lambda^\star$ becomes zero.
In this case, we can say that the interfering channels are perfectly
aligned in $(N_r-1)$-dimensional subspace in $\mathbb{C}^{N_r}$.
Note that $S_2 (\mathbf{c}^\star , 0)$ is an $(N_r-1)$-dimensional
subspace orthogonal to $\mathbf{c}^\star$, and $S_2
(\mathbf{c}^\star , 1)$ is the $N_r$-dimensional complex
hypersphere, $S_0$.
When $\lambda^\star$ is the smaller, the vectors are the more
aligned in the $(N_r-1)$-dimensional subspace, $\mathcal{U}$.
Thus, we will use $\lambda^\star $ as an \emph{interference
alignment measure} to quantify how much the interfering channels are
closely aligned in an $(N_r-1)$-dimensional subspace, i.e.,
\begin{align}
    \mathfrak{q}(\mathbf{g}_1, \ldots, \mathbf{g}_{K-1})
    &= \min_{\Vert \mathbf{c} \Vert=1}\max_{1\le k\le K-1}
        \vert \mathbf{c} ^\dagger \mathbf{g}_k\vert^2
    \label{eqn:AM1}\\
    &= \lambda^\star   \quad (\lambda^\star \in [0,1]).\label{eqn:lambda*}
\end{align}
In other words, we use the mini-max distance of the interfering
channels from an $(N_r-1)$-dimensional subspace. In Fig.
\ref{fig:IAM_concept}, the interference alignment measure is
geometrically represented.
The more the interfering channels are aligned, the smaller the
interference alignment measure becomes.

Since the interference alignment measure is obtained from the
optimization problem \eqref{eqn:optimization_prob}, the exact
distribution is difficult to find.
Instead, we obtain the lower bound for the cumulative distribution
function (CDF) of the interference alignment measure in the
following lemma.

%%%%%%%%%%%%%%%%%%%%%%%%%%%%%%%%%%%%%%%%%%%%%%%%%%%%%%%%%%%%%%%%%%%
% % % % % % % % % % % % % % % % % % % % % % % % % % % % % % % % % %
%%%%%%%%%%%%%%%%%%%%%%%%%%%%%%%%%%%%%%%%%%%%%%%%%%%%%%%%%%%%%%%%%%%
\begin{lemma}\label{lemma:prob_bar_D_M} When $K > N_r$, the
probability that the interference alignment measure
$\mathfrak{q}(\mathbf{g}_1, \ldots, \mathbf{g}_{K-1})$ is smaller
than $\lambda\in [0,1]$ is lower bounded on
\begin{align}
 \PR{\mathfrak{q}(\mathbf{g}_1, \ldots, \mathbf{g}_{K-1}) \le \lambda}
 \ge \left(1-(1-\lambda)^{N_r-1}\right)^{K-N_r}.
 \label{eqn:AM1_prob}
\end{align}
\end{lemma}
%%%%%%%%%%%%%%%%%%%%%%%%%%%%%%%%%%%%%%%%%%%%%%%%%%%%%%%%%%%%%%%%%%%
\begin{proof} We consider two events:
\begin{align}
&\mathrm{(E1)}:~~
    \mathfrak{q} (\mathbf{g}_1, \ldots, \mathbf{g}_{K-1}) \le
    \lambda \NNL
&\mathrm{(E2)}:~~
    S_2(\bar{\mathbf{c}}, \lambda)
    \supset \{\mathbf{g}_1, \ldots, \mathbf{g}_{K-1}\},  \NN
\end{align}
where $\bar{\mathbf{c}}$ is the $N_r$-dimensional unit vector such
that $\bar{\mathbf{c}} \perp\{\mathbf{g}_1, \ldots, \mathbf{g}
_{N_r-1}\}$.
By the definition of the interference alignment measure given in
\eqref{eqn:AM1}, $\rm{(E1)}$ is true whenever $\mathrm{(E2)}$ is
true, equivalently, $\PR{\mathrm{(E1})} \ge \PR{\mathrm{(E2)}}$.
The probability of $\mathrm{(E2)}$ is obtained by
\begin{align}
 \PR{\mathrm{(E2)}}
 &=\PR{S_2(\bar{\mathbf{c}}, \lambda) \supset \{\mathbf{g}_1, \ldots,
    \mathbf{g}_{K-1}\}}\NNL
 &\stackrel{(a)}{=}
   \PR{S_2(\bar{\mathbf{c}}, \lambda) \supset \{\mathbf{g}_{N_r}, \ldots,
    \mathbf{g}_{K-1}\}}\NNL
 &\stackrel{(b)}{=}
    \left(1-(1-\lambda)^{N_r-1}\right)^{K-N_r},
\end{align}
where the equality $(a)$ is from the definition of
$\bar{\mathbf{c}}$ such that $\bar{\mathbf{c}} \perp\{\mathbf{g}_1,
\ldots, \mathbf{g} _{N_r-1}\}$.
Also, the equality $(b)$ holds from Lemma \ref{lemma:p_n_lambda} and
from the fact that $\bar{\mathbf{c}}$ is independent of
$\{\mathbf{g}_{N_r}, \ldots, \mathbf{g}_{K-1}\}$.
Thus, we obtain
\begin{align}
 \PR{\mathrm{(E1)}} \ge
 \left(1-(1-\lambda)^{N_r-1}\right)^{K-N_r}.
\end{align}
\end{proof}
\vspace{.1in}

%%%%%%%%%%%%%%%%%%%%%%%%%%%%%%%%%%%%%%%%%%%%%%%%%%%%%%%%%%%%%%%%%%%
% % % % % % % % % % % % % % % % % % % % % % % % % % % % % % % % % %
%%%%%%%%%%%%%%%%%%%%%%%%%%%%%%%%%%%%%%%%%%%%%%%%%%%%%%%%%%%%%%%%%%%
\subsection{Achievable Value of the Interference Alignment Measure via User Selection}

The remaining question is how much we can reduce the interference
alignment measure via user selection.
In the first user group, the $n$th user has $K-1$ interfering
channels, $\mathbf{h}_{n,2}, \ldots, \mathbf{h}_{n,K}$.
The interference alignment measure at the $n$th user can be written
by
\begin{align}
 \mathfrak{q} \big( \tilde{\mathbf{h}} _{n,2},
   \ldots, \tilde{\mathbf{h}}_{n,K}\big),
   \label{eqn:AM_user}
\end{align}
where $\tilde{\mathbf{h}}_{n,k}$ is the normalized interfering
channel, i.e., $\tilde{\mathbf{h}}_{n,k} = \h{n,k} /
\Vert\h{n,k}\Vert$.
Thus, the achievable smallest interference alignment measure via
user selection is given by
\begin{align}
 \min_n~\mathfrak{q} \big( \tilde{\mathbf{h}} _{n,2},
   \ldots, \tilde{\mathbf{h}}_{n,K}\big).
   \label{eqn:AM_smallest}
\end{align}
Obviously, the smallest interference alignment measure will decrease
as the number of users increases.
In the following lemma, we find the relationship between
\eqref{eqn:AM_smallest} and the number of total users (i.e., $N$).

%%%%%%%%%%%%%%%%%%%%%%%%%%%%%%%%%%%%%%%%%%%%%%%%%%%%%%%%%%%%%%%%%
%\vspace{.1in}
\begin{lemma}\label{lemma:lambda_n1}
When there are $N$ users, the expectation of the smallest
interference alignment measure is upper bounded on
\begin{align}
 \mathbb{E}\left[\min_{n}~\mathfrak{q} \big( \tilde{\mathbf{h}} _{n,2},
   \ldots, \tilde{\mathbf{h}}_{n,K}\big)\right]
   < N^{-\frac{1}{K-N_r}}.\label{eqn:AM1_bound}
\end{align}

\end{lemma}
\vspace{.1in}
%%%%%%%%%%%%%%%%%%%%%%%%%%%%%%%%%%%%%%%%%%%%%%%%%%%%%%%%%%%%%%%%%
\begin{proof}
The complementary CDF of \eqref{eqn:AM_smallest} is bounded on
\begin{align}
 \PR{\min_n\mathfrak{q} \big( \tilde{\mathbf{h}} _{n,2},
   \ldots, \tilde{\mathbf{h}}_{n,K}\big) \ge \lambda}
 %\NNL
 &\qquad=\PR{\mathfrak{q} \big( \tilde{\mathbf{h}} _{n,2},
   \ldots, \tilde{\mathbf{h}}_{n,K}\big) \ge \lambda \textrm{~for all~} n}\NNL
 &\qquad=\prod_{n=1}^N\PR{
    \mathfrak{q} \big( \tilde{\mathbf{h}} _{n,2},
   \ldots, \tilde{\mathbf{h}}_{n,K}\big)
   \ge \lambda}\NNL
 &\qquad=\left(1-\PR{
    \mathfrak{q} \big( \tilde{\mathbf{h}} _{n,2},
   \ldots, \tilde{\mathbf{h}}_{n,K}\big)
   \le \lambda}\right)^N\NNL
 &\qquad\stackrel{(a)}{<} \big[1- (1-(1 - \lambda)^{N_r-1})^{K-N_r}\big]^N,
    \label{eqn:prob_bound}
\end{align}
where $\lambda\in[0,1]$, and the inequality $(a)$ holds from Lemma
\ref{lemma:prob_bar_D_M}.
Using this bound, we obtain \eqref{eqn:AM1_bound} as
\begin{align}
 \mathbb{E} \left[\min_{n}~\mathfrak{q} \big(
    \tilde{\mathbf{h}} _{n,2}, \ldots, \tilde{\mathbf{h}} _{n,K} \big)
    \right]
    &=\int_0^1 \PR{\min_n \mathfrak{q} \big( \tilde{\mathbf{h}} _{n,2},
   \ldots, \tilde{\mathbf{h}}_{n,K}\big)\ge \lambda} d \lambda \NNL
    &\le \int_0^1 \big[1- (1 - (1- \lambda) ^{N_r-1})^{K-N_r}\big]^N d \lambda\NNL
    &\stackrel{(a)}{\le} \int_0^1 \big[1- (1 - (1- \lambda))^{K-N_r}\big]^N d \lambda\NNL
    &\stackrel{(b)}{=} \frac{1}{K-N_r}\beta\left(\frac{1}{K-N_r}, N+1\right)\NNL
    &\stackrel{(c)}{=} \frac{\Gamma\left(1+\frac{1}{K-N_r}\right)\Gamma(N+1)}
        {\Gamma\left(N+ 1 + \frac{1}{K-N_r}\right)}
        \NNL
%    &\stackrel{(d)}{<}
%        \frac{(N+1)^{1-\frac{1}{K-N_r}}}{\left(N+\tfrac{1}{K-N_r}\right)}
%        \NNL
    &\stackrel{(d)}{<} N^{-\frac{1}{K-N_r}},\NN
\end{align}
where the inequality $(a)$ is due to $(1-\lambda)^{N_r-1} \le
(1-\lambda)$ for $0\le \lambda\le 1$, and the equality $(b)$ holds
from the representation of beta function \cite[p.324]{GR2007}
\begin{align}
 \int_0^1 x^{p-1}(1-x^q)^{r-1}dx
   =\frac{1}{q}\beta\left(\frac{p}{q},r\right).\NN
\end{align}
The equality $(c)$ comes from the definition of the beta function
$\beta(p,q) = \Gamma(p)\Gamma(q)/\Gamma(p+q)$ and the property of
the Gamma function $\Gamma(p+1)=p\Gamma(p)$.
In the right-hand-side of the equality $(c)$, it holds
$\Gamma(1+\frac{1}{K-N_r})<1$ because $0 < \Gamma(x) < 1$ for $1< x
< 2$. Also, it is satisfied that
\begin{align}
\frac{\Gamma(N+1)}
        {\Gamma\left(N+ 1 + \frac{1}{K-N_r}\right)}
    &\stackrel{(e)}{<}
        \left(N+1+\frac{1}{K-N_r}\right)^{-\frac{1}{K-N_r}}
        \NNL
    &< N^{-\frac{1}{K-N_r}},\NN
\end{align}
where $(e)$ is from the Gautschi's inequality \cite{G1959} given by
\begin{align}
 \frac{\Gamma(x+s)}{\Gamma(x+1)}<(x+1)^{s-1}, ~~~ \textrm{for~~}x>0,~~0<s<1,\NN
\end{align}
with $x=N+\frac{1}{K-N_r}$ and $s=1-\frac{1}{K-N_r}$. Thus, the
inequality $(d)$ holds.
\end{proof}%\vspace{.1in}
%%%%%%%%%%%%%%%%%%%%%%%%%%%%%%%%%%%%%%%%%%%%%%%%%%%%%%%%%%%%%%%%%

%%%%%%%%%%%%%%%%%%%%%%%%%%%%%%%%%%%%%%%%%%%%%%%%%%%%%%%%%%%%%%%%%
\begin{comment}
\begin{remark}
The OIA scheme proposed scheme in \cite{LC2010, LC2011, LC2011_2}
can be extended to our system model by using the interference
alignment measure defined in \eqref{eqn:AM_user}.
%
Each user feeds the interference alignment measure to the
transmitter, and the transmitter selects the user having the
smallest interference alignment measure as in
\eqref{eqn:AM_smallest}.
\end{remark}
\end{comment}
%%%%%%%%%%%%%%%%%%%%%%%%%%%%%%%%%%%%%%%%%%%%%%%%%%%%%%%%%%%%%%%%%

%%%%%%%%%%%%%%%%%%%%%%%%%%%%%%%%%%%%%%%%%%%%%%%%%%%%%%%%%%%%%%%%%%%
% % % % % % % % % % % % % % % % % % % % % % % % % % % % % % % % % %
%%%%%%%%%%%%%%%%%%%%%%%%%%%%%%%%%%%%%%%%%%%%%%%%%%%%%%%%%%%%%%%%%%%
\JHnewpage
\section{Optimal Exploitation of Multiuser Diversity for the Target DoF}

In this section, we derive the optimal strategies of exploiting
multiuser diversity for the target DoF $d$. We first decompose the
target DoF $d$ into the DoF gain term $d_1$ and the DoF loss term
$d_2$ such that $d=d_1-d_2$, and find the required user scalings for
$d_1$ and $d_2$, respectively. Then, the optimal target DoF
achieving strategy is derived by determining the optimal combination
$(d_1^\star,d_2^\star)$ which requires the minimum user scaling for
the target DoF $d$.

%%%%%%%%%%%%%%%%%%%%%%%%%%%%%%%%%%%%%%%%%%%%%%%%%%%%%%%%%%%%%%%%%%%
% % % % % % % % % % % % % % % % % % % % % % % % % % % % % % % % % %
%%%%%%%%%%%%%%%%%%%%%%%%%%%%%%%%%%%%%%%%%%%%%%%%%%%%%%%%%%%%%%%%%%%
\subsection{Required User Scaling to Reduce the DoF Loss Term}

In this subsection, we find the required user scaling to reduce the
DoF loss. Via user selection, the rate loss term given in
\eqref{eqn:Rloss} can be minimized by
\begin{align}
 \mathbb{E}
 \left[\min_{n,\mathbf{v}_n}~\log_2\left(1+  P \sum_{k=2}^{K}
    \vert \mathbf{v}_{n}^\dagger \h{n,k} \vert^2\right)
    \right].  \label{eqn:Rloss_min}
\end{align}
This value is upper bounded on
\begin{align}
 \mathbb{E}
 \left[\min_{n,\mathbf{v}_n}~\log_2\left(1+  P \sum_{k=2}^{K}
    \vert \mathbf{v}_{n}^\dagger \h{n,k} \vert^2\right)
    \right]
 &\stackrel{(a)}{=}
 \mathbb{E}_{\Vert\mathbf{h}\Vert, \tilde{\mathbf{h}}}
 \left[\min_{n,\mathbf{v}_n}
 ~\log_2\left(1+  P \sum_{k=2}^{K}
    \Vert\mathbf{h}_{n,k}\Vert^2\vert \mathbf{v}_{n}^\dagger \tilde{\mathbf{h}}_{n,k} \vert^2\right)
    \right]\NNL
 &\stackrel{(b)}{\le}
 \mathbb{E}_{\tilde{\mathbf{h}}}
 \left[\min_{n,\mathbf{v}_n}
 ~\mathbb{E}_{\Vert\mathbf{h}\Vert}\log_2\left(1+  P \sum_{k=2}^{K}
    \Vert\mathbf{h}_{n,k}\Vert^2\vert \mathbf{v}_{n}^\dagger \tilde{\mathbf{h}}_{n,k} \vert^2\right)
    \right]\NNL
 &\stackrel{(c)}{\le}
 \mathbb{E}_{\tilde{\mathbf{h}}}
 \left[\min_{n,\mathbf{v}_n}
 ~\log_2\left(1+  N_rP \sum_{k=2}^{K}
    \vert \mathbf{v}_{n}^\dagger \tilde{\mathbf{h}}_{n,k} \vert^2\right)
    \right]\NNL
 &\stackrel{(d)}{\le}
 \mathbb{E}_{\tilde{\mathbf{h}}}
 \left[\min_{n}
 ~\log_2\left(1+  N_rP
    (K-1)\mathfrak{q}(\tilde{\mathbf{h}}_{n,2}, \ldots, \tilde{\mathbf{h}}_{n,K})
  \right)\right]\NNL
 &\stackrel{(e)}{\le}
 ~\log_2\left(1+  N_rP (K-1)
 \mathbb{E}_{\tilde{\mathbf{h}}}\left[
    \min_{n} \mathfrak{q}(\tilde{\mathbf{h}}_{n,2}, \ldots, \tilde{\mathbf{h}}_{n,K})
  \right]\right)\NNL
 &\stackrel{(f)}{\le}
 ~\log_2\left(1+  N_rP (K-1) N^{-\frac{1}{K-N_R}}\right),
 \label{eqn:Rloss_bound}
\end{align}
where the equality $(a)$ is obtained by decomposing the channel
vector into direction and magnitude independent of each other such
that $\h{n,k}=\Vert\h{n,k}\Vert \tilde{\mathbf{h}}_{n,k}$. The
inequality $(b)$ holds because the minimum of the average is larger
than the average of the minimum. The inequality $(c)$ is from the
Jensen's inequality and $\mathbb{E}\Vert\h{n,k}\Vert^2=N_r$. Also,
the inequality $(d)$ holds from the fact that
\begin{align}
 \min_{\mathbf{v}_n}
    \left[\sum_{k=2}^{K} \vert \mathbf{v}_{n}^\dagger
    \tilde{\mathbf{h}}_{n,k} \vert^2\right]
 &\le \min_{\mathbf{v}_n} \left[(K-1) \max_{2\le k\le K}
    \vert \mathbf{v}_n^\dagger \tilde{\mathbf{h}}_{n,k}\vert^2
    \right]\NNL
 &= (K-1) \mathfrak{q}
    (\tilde{\mathbf{h}}_{n,2}, \ldots, \tilde{\mathbf{h}}_{n,K}),
\end{align}
where $\mathfrak{q}(\tilde{\mathbf{h}}_{n,2}, \ldots,
\tilde{\mathbf{h}}_{n,K})$ is the interference alignment measure at
the user $n$ given in \eqref{eqn:AM1}.
The inequality $(e)$ is from the Jensen's inequality, and the
inequality $(f)$ holds from Lemma \ref{lemma:lambda_n1}.
We obtain the following theorem.

%%%%%%%%%%%%%%%%%%%%%%%%%%%%%%%%%%%%%%%%%%%%%%%%%%%%%%%%%%%%%%%%%%%
% % % % % % % % % % % % % % % % % % % % % % % % % % % % % % % % % %
%%%%%%%%%%%%%%%%%%%%%%%%%%%%%%%%%%%%%%%%%%%%%%%%%%%%%%%%%%%%%%%%%%%
\begin{theorem}\label{theorem:N_scaling_DoF_loss}
We can obtain the DoF loss term $d_2 \in [0, 1]$ when the number of
users in each group scales as
\begin{align}
 N\propto P^{(1-d_2)(K-N_r)}.\NN
\end{align}
\end{theorem}
%%%%%%%%%%%%%%%%%%%%%%%%%%%%%%%%%%%%%%%%%%%%%%%%%%%%%%%%%%%%%%%%%
\begin{proof}
To obtain the DoF loss term $d_2$,
it is enough to make \eqref{eqn:Rloss_bound} satisfying
\begin{align}
 \DoF{\log_2\left(1+  N_rP (K-1) N^{-\frac{1}{K-N_R}}\right)} = d_2,
\end{align}
which is achieved if $N \propto P^{(1-d_2)(K-N_r)}$.
\end{proof}
A tighter upper bound of the rate loss term than
\eqref{eqn:Rloss_bound} could exist, but the derived upper bound in
\eqref{eqn:Rloss_bound} enables us to compare the increasing speeds
of the transmit power and the required number of users, which is the
crucial factor of DoF calculation. The scaling law of the required
number of users obtained from \eqref{eqn:Rloss_bound}, which is
derived in Theorem \ref{theorem:N_scaling_DoF_loss}, is enough to
find the optimal target DoF achieving strategy as shown in Section
\ref{sec:DoF_achieving_strategy}.

%%%%%%%%%%%%%%%%%%%%%%%%%%%%%%%%%%%%%%%%%%%%%%%%%%%%%%%%%%%%%%%%%%%
% % % % % % % % % % % % % % % % % % % % % % % % % % % % % % % % % %
%%%%%%%%%%%%%%%%%%%%%%%%%%%%%%%%%%%%%%%%%%%%%%%%%%%%%%%%%%%%%%%%%%%

%%%%%%%%%%%%%%%%%%%%%%%%%%%%%%%%%%%%%%%%%%%%%%%%%%%%%%%%%%%%%%%%%%%
% % % % % % % % % % % % % % % % % % % % % % % % % % % % % % % % % %
%%%%%%%%%%%%%%%%%%%%%%%%%%%%%%%%%%%%%%%%%%%%%%%%%%%%%%%%%%%%%%%%%%%
\subsection{Required User Scaling to Increase the DoF Gain Term}
\label{sec:DoF_gain_scaling}

We also find the required user scaling to increase the DoF gain
term.
From the definition of the rate gain term given in
\eqref{eqn:Rgain}, the maximum rate gain term obtained by user
selection is
\begin{align}
 \mathbb{E}\left[
    \max_{n, \mathbf{v}_n} \log_2\left(1+  P \sum_{k=1}^{K}
    \vert \mathbf{v}_{n}^\dagger \h{n,k} \vert^2\right)\right].
    \label{eqn:Rgain_max}
\end{align}
This value is lower bounded on
\begin{align}
\mathbb{E}\left[
    \max_{n, \mathbf{v}_n} \log_2\left(1+  P \sum_{k=1}^{K}
    \vert \mathbf{v}_{n}^\dagger \h{n,k} \vert^2\right)\right]
 &\ge\mathbb{E}\left[
    \max_{n, \mathbf{v}_n} \log_2\left(1+  P
    \vert \mathbf{v}_{n}^\dagger \h{n,1} \vert^2\right)\right]\NNL
 &= \mathbb{E}\left[
    \max_{n} \log_2\left(1+  P \Vert \h{n,1}
    \Vert^2\right)\right],\label{eqn:Rgain_lower_bound}
\end{align}
and upper bounded on
\begin{align}
\mathbb{E}\left[
    \max_{n, \mathbf{v}_n} \log_2\left(1+  P \sum_{k=1}^{K}
    \vert \mathbf{v}_{n}^\dagger \h{n,k} \vert^2\right)\right]
 &\le \mathbb{E}\left[
    \max_{n} \log_2\left(1+  P \sum_{k=1}^{K}
    \Vert \h{n,k} \Vert^2\right)\right]\NNL
 &\le  \mathbb{E}\left[
    \max_{n} \log_2\left(1+  PK \max_k
    \Vert \h{n,k} \Vert^2\right)\right].\label{eqn:Rgain_upper_bound}
\end{align}
Since all $\Vert\h{n,k}\Vert^2$ are i.i.d. $\chi^2(2N_r)$ random
variables, for sufficiently large $N$, the bounds
\eqref{eqn:Rgain_lower_bound} and \eqref{eqn:Rgain_upper_bound} acts
like \cite{SH2007}
\begin{align}
 \mathbb{E}\left[
    \max_{n} \log_2\left(1+  P \Vert \h{n,1}
    \Vert^2\right)\right]
    &\sim \log_2(1+P\log N) \NNL
 \mathbb{E}\left[
    \max_{n,k} \log_2\left(1+  PK
    \Vert \h{n,k} \Vert^2\right)\right]
    &\sim \log_2(1+PK\log (KN)).\NN
\end{align}
Thus, when both $N$ and $P$ are large enough, \eqref{eqn:Rgain_max}
act like $\log_2(P\log N)$, i.e,
\begin{align}
 \lim_{P\to\infty \atop N\to\infty}
 \mathbb{E}\left[
    \max_{n, \mathbf{v}_n} \log_2\left(1+  P \sum_{k=1}^{K}
    \vert \mathbf{v}_{n}^\dagger \h{n,k} \vert^2\right)\right]
 \sim \log_2(P\log N). \label{eqn:Rgain_asymptotic}
\end{align}
Therefore, we establish the following theorem.

%%%%%%%%%%%%%%%%%%%%%%%%%%%%%%%%%%%%%%%%%%%%%%%%%%%%%%%%%%%%%%%%%%%
% % % % % % % % % % % % % % % % % % % % % % % % % % % % % % % % % %
%%%%%%%%%%%%%%%%%%%%%%%%%%%%%%%%%%%%%%%%%%%%%%%%%%%%%%%%%%%%%%%%%%%
\begin{theorem}\label{theorem:N_scaling_DoF_gain}
The DoF gain term $d_1 \in[1,\infty)$ is achievable when the number
of users in each group scales as
\begin{align}
 N\propto e^{P^{(d_1-1)}}.\NN
\end{align}
\end{theorem}
%%%%%%%%%%%%%%%%%%%%%%%%%%%%%%%%%%%%%%%%%%%%%%%%%%%%%%%%%%%%%%%%%
\begin{proof}
We use \eqref{eqn:Rgain_asymptotic}.
By setting
\begin{align}
 \DoF{\log_2(P\log N)} = d_1,
\end{align}
we obtain the required user scaling for the DoF gain term $d_1$
given by $N\propto e^{P^{(d_1-1)}}$.
\end{proof}
%%%%%%%%%%%%%%%%%%%%%%%%%%%%%%%%%%%%%%%%%%%%%%%%%%%%%%%%%%%%%%%%%%%
% % % % % % % % % % % % % % % % % % % % % % % % % % % % % % % % % %
%%%%%%%%%%%%%%%%%%%%%%%%%%%%%%%%%%%%%%%%%%%%%%%%%%%%%%%%%%%%%%%%%%%

%%%%%%%%%%%%%%%%%%%%%%%%%%%%%%%%%%%%%%%%%%%%%%%%%%%%%%%%%%%%%%%%%%%
% % % % % % % % % % % % % % % % % % % % % % % % % % % % % % % % % %
%%%%%%%%%%%%%%%%%%%%%%%%%%%%%%%%%%%%%%%%%%%%%%%%%%%%%%%%%%%%%%%%%%%
\subsection{Target DoF Achieving Strategy}\label{sec:DoF_achieving_strategy}

In Theorem \ref{theorem:N_scaling_DoF_loss} and Theorem
\ref{theorem:N_scaling_DoF_gain}, we found the required user
scalings for the DoF loss term $d_2$ and the DoF gain term $d_1$,
respectively. In this subsection, we find the optimal target DoF
achieving strategy which requires the minimum user scaling. We start
with the following theorem.
%
%%%%%%%%%%%%%%%%%%%%%%%%%%%%%%%%%%%%%%%%%%%%%%%%%%%%%%%%%%%%%%%%%%%
% % % % % % % % % % % % % % % % % % % % % % % % % % % % % % % % % %
%%%%%%%%%%%%%%%%%%%%%%%%%%%%%%%%%%%%%%%%%%%%%%%%%%%%%%%%%%%%%%%%%%%
\begin{theorem}\label{theorem:scaling1}
For the target DoF up to one, the whole multiuser dimensions should
be devoted to minimizing the DoF loss caused by interfering signals.
The optimal DoF achieving strategy for the target DoF $d \in [0, 1]$
is $(d_1^\star, d_2^\star) = (1, 1-d)$, and the corresponding
sufficient user scaling is
\begin{align}\label{eqn:scaling(d<1)}
    N\propto P^{d(K-N_r)}.
\end{align}
\end{theorem}
%%%%%%%%%%%%%%%%%%%%%%%%%%%%%%%%%%%%%%%%%%%%%%%%%%%%%%%%%%%%%%%%%
\begin{proof}
In Theorem \ref{theorem:N_scaling_DoF_loss}, we have shown that the
target DoF $d \in [0,1]$ is achievable by reducing the DoF loss term
with the user scaling $N\propto P^{d(K-N_r)}$.
On the other hand, this user scaling cannot increase the DoF gain
term.
Substituting $N\propto P^{d(K-N_r)}$ into
\eqref{eqn:Rgain_asymptotic} the DoF gain term $d_1$ becomes
\begin{align}
 \DoF{\log_2\left(P \log( P^{d(K-N_r)})\right)} = 1.
 \label{eqn:DoF_gain_term1}
\end{align}
which is the same as when there is a fixed number of users as
described in \eqref{eqn:DoF_fixedN}.
That is, any other combinations $(d_1,d_2) = (1+\Delta, 1-d +
\Delta)$ which achieve the target DoF $d$ requires larger user
scaling than $N\propto P^{d(K-N_r)}$, where $\Delta >0$ since
$d_1>1$.
Therefore, the optimal target DoF achieving strategy is given by
$(d_1^\star, d_2^\star) = (1, 1-d)$, and the sufficient user scaling
is $N\propto P^{d(K-N_r)}$.
\end{proof}
%%%%%%%%%%%%%%%%%%%%%%%%%%%%%%%%%%%%%%%%%%%%%%%%%%%%%%%%%%%%%%%%%%%
% % % % % % % % % % % % % % % % % % % % % % % % % % % % % % % % % %
%%%%%%%%%%%%%%%%%%%%%%%%%%%%%%%%%%%%%%%%%%%%%%%%%%%%%%%%%%%%%%%%%%%
\JHnewpage

Now, we derive the target DoF achieving strategy when the target DoF
$d$ is greater than one. To find the optimal DoF achieving strategy,
we firstly find the sufficient user scaling for an arbitrary
strategy $(d_1, d_2)$ achieving  DoF $d ~(=d_1-d_2>1)$. Then, we
show that the optimal target DoF achieving strategy for the target
DoF $d ~(>1)$ is $(d_1^\star, d_2^\star)=(d,0)$.

%%%\begin{figure}[!t]
%%%\centering
%%%  \includegraphics[width=.8777\columnwidth]{./figures/user_selection_2stage_new.eps}
%%%  \caption{Two-stage user selection scheme.}
%%%  \label{fig:user_selection_2stage}
%%%\end{figure}

%%%%%%%%%%%%%%%%%%%%%%%%%%%%%%%%%%%%%%%%%%%%%%%%%%%%%%%%%%%%%%%%%%%
% % % % % % % % % % % % % % % % % % % % % % % % % % % % % % % % % %
%%%%%%%%%%%%%%%%%%%%%%%%%%%%%%%%%%%%%%%%%%%%%%%%%%%%%%%%%%%%%%%%%%%
\begin{lemma}\label{lemma:general_scaling_d>1}
For the target DoF $d~(>1)$, the sufficient user scaling for an
arbitrary strategy $(d_1, d_2)$ achieving DoF $d~(=d_1-d_2)$ is
given by
\begin{align}
    N \propto e^{P^{(d_1-1)}} P^{(1-d_2)(K-N_r)},
    \label{eqn:general_scaling(d>1)}
\end{align}
where $d_1>1$ and $d_2\in[0,1]$.
\end{lemma}
%%%%%%%%%%%%%%%%%%%%%%%%%%%%%%%%%%%%%%%%%%%%%%%%%%%%%%%%%%%%%%%%%%%
\begin{proof}
As a target DoF achieving scheme, we consider a two-stage user
selection scheme; the first stage is to increase the DoF gain term,
and the second stage is to decrease the DoF loss term.
The considered two stage user selection strategy is illustrated in
Fig. \ref{fig:user_selection_2stage}. We randomly divide total $N$
users into $N_2$ subgroups having $N_1$ users each such that $N_1
N_2=N$.
Then, the user selection in each stage is performed as follows.
\begin{itemize}
\item Stage 1: In each subgroup, a single user having the largest
channel gain is selected among $N_1$ users. As a result, we have
$N_2$ selected users after Stage 1.

\item Stage 2: Among the $N_2$ users, the transmitter selects a single
user to minimize the DoF loss term.

\end{itemize}
In Stage 1, the DoF gain term $d_1$ is obtained at each selected
user when
\begin{align}
 N_1 \propto e^{P^{(d_1-1)}}
\end{align}
as stated in Theorem \ref{theorem:N_scaling_DoF_gain}.
In Stage 2, we can make the DoF loss term $d_2$ when
\begin{align}
    N_2\propto P^{(1-d_2)(K-N_r)}
\end{align}
as shown in Theorem \ref{theorem:N_scaling_DoF_loss}.
Thus, the target DoF $d~(>1)$ with the strategy $(d_1, d_2)$ such
that $d=d_1-d_2$ can be obtained by the user scaling $N_1N_2$, which
is given in \eqref{eqn:general_scaling(d>1)}.
\end{proof}
%%%%%%%%%%%%%%%%%%%%%%%%%%%%%%%%%%%%%%%%%%%%%%%%%%%%%%%%%%%%%%%%%%%
% % % % % % % % % % % % % % % % % % % % % % % % % % % % % % % % % %
%%%%%%%%%%%%%%%%%%%%%%%%%%%%%%%%%%%%%%%%%%%%%%%%%%%%%%%%%%%%%%%%%%%

From Lemma \ref{lemma:general_scaling_d>1}, we obtain the optimal
DoF achieving strategy for the target DoF $d~ (>1)$ in following
theorem.

%%%%%%%%%%%%%%%%%%%%%%%%%%%%%%%%%%%%%%%%%%%%%%%%%%%%%%%%%%%%%%%%%%%
% % % % % % % % % % % % % % % % % % % % % % % % % % % % % % % % % %
%%%%%%%%%%%%%%%%%%%%%%%%%%%%%%%%%%%%%%%%%%%%%%%%%%%%%%%%%%%%%%%%%%%
\begin{theorem}\label{theorem:scaling2}
The optimal target DoF achieving strategy for $d\in [1,\infty)$ is
to increase the DoF gain term to $d$ and to perfectly eliminate the
DoF loss, i.e., $(d_1^\star, d_2^\star) = (d, 0)$. Consequently, the
sufficient user scaling for target DoF $d~(>1)$ becomes
\begin{align}
 N \propto e^{P^{(d-1)}} P^{(K-N_r)}.
 \label{eqn:scaling(d>1)}
\end{align}
\end{theorem}
%%%%%%%%%%%%%%%%%%%%%%%%%%%%%%%%%%%%%%%%%%%%%%%%%%%%%%%%%%%%%%%%%%%
\begin{proof}
The proof is similar to that of Theorem \ref{theorem:scaling1}.
From Lemma \ref{lemma:general_scaling_d>1}, we can obtain the target
DoF $d~(>1)$ by the strategy $(d, 0)$ with the sufficient user
scaling given in \eqref{eqn:scaling(d>1)}.
However, this scaling cannot increase the DoF gain term larger than
$d$ even when the user scaling is only used to increase the DoF gain
term. Substituting \eqref{eqn:scaling(d>1)} into
\eqref{eqn:Rgain_asymptotic}, we still have
\begin{align}
 \DoF{\log_2\left(P \log( e^{P^{(d-1)}} P^{(K-N_r)})\right)} = d.
 \label{eqn:still_d}
\end{align}
This implicates that the user scaling given in
\eqref{eqn:scaling(d>1)} is sufficient for the strategy $(d, 0)$ but
not enough for other strategies $(d+\Delta, 1)$ as well as
$(d+\Delta, \Delta)$ which requires the higher user scaling than
that of $(d+\Delta, 1)$, where $\Delta \in(0, 1]$. Therefore, the
optimal strategy for the target DoF $d~(>1)$ becomes $(d_1^\star,
d_2^\star)=(d,0)$.
\end{proof}
%%%%%%%%%%%%%%%%%%%%%%%%%%%%%%%%%%%%%%%%%%%%%%%%%%%%%%%%%%%%%%%%%%%
% % % % % % % % % % % % % % % % % % % % % % % % % % % % % % % % % %
%%%%%%%%%%%%%%%%%%%%%%%%%%%%%%%%%%%%%%%%%%%%%%%%%%%%%%%%%%%%%%%%%%%

%%%\begin{figure}[!t]
%%%\centering
%%%  \includegraphics[width=.8777\columnwidth]{./figures/DoF_achieving_strategy.eps}
%%%  \caption{The optimal DoF achieving strategy $(d_1^\star, d_2^\star )$
%%%    for the target DoF $d$}
%%%  \label{fig:DoF_achieving_strategy}
%%%\end{figure}

In Fig. \ref{fig:DoF_achieving_strategy}, the optimal DoF achieving
strategy $(d_1^\star, d_2^\star)$ is plotted according to the target
DoF $d ~(=d_1-d_2)$.

%%%%%%%%%%%%%%%%%%%%%%%%%%%%%%%%%%%%%%%%%%%%%%%%%%%%%%%%%%%%%%%%%%%
% % % % % % % % % % % % % % % % % % % % % % % % % % % % % % % % % %
%%%%%%%%%%%%%%%%%%%%%%%%%%%%%%%%%%%%%%%%%%%%%%%%%%%%%%%%%%%%%%%%%%%
\JHnewpage
\section{Practical User Selection Schemes}

%%%%%%%%%%%%%%%%%%%%%%%%%%%%%%%%%%%%%%%%%%%%%%%%%%%%%%%%%%%%%%%%%%%
% % % % % % % % % % % % % % % % % % % % % % % % % % % % % % % % % %
%%%%%%%%%%%%%%%%%%%%%%%%%%%%%%%%%%%%%%%%%%%%%%%%%%%%%%%%%%%%%%%%%%%
In this section, we discuss how the optimal target DoF achieving
strategy can be realized by practical user selection schemes.
The practical schemes considered in this section require no
cooperation and no information exchange among the transmitters.
For practical scenarios, we assume that each user has knowledge of
channel state information (CSI) of the direct channel and the
covariance matrix of the received signal without explicit knowledge
of CSI of the interfering channels. That is, the $n$th user knows
CSI of its own desired channel $\mathbf{h}_{n,1}$ and the covariance
matrix of the received signal $\mathbb{E} [\mathbf{y}_n \mathbf{y}_n
^\dagger] = \mathbf{I}_{N_r} + P\sum_{k=1}^K \mathbf{h}_{n,k}
\mathbf{h}_{n,k} ^\dagger$. From these values, the user $n$ easily
obtains the interference covariance matrix denoted by $\mathbf{R}_n
\triangleq P\sum_{k=2}^K \mathbf{h}_{n,k} \mathbf{h}_{n,k} ^\dagger$
such as
\begin{align}
  \mathbf{R}_n
  &=\mathbb{E}[\mathbf{y}_n \mathbf{y}_n ^\dagger]
    - P\mathbf{h}_{n,1}\mathbf{h}_{n,1}^\dagger
    -\mathbf{I}_{N_r}.
\end{align}
Therefore, the achievable rate at the first transmitter given in
\eqref{eqn:C_n} can be rewritten by
\begin{align}
 \R
 &\triangleq \mathbb{E}\log_2
    \left(1+
    \frac{ P\vert \mathbf{v}_{n^*}^\dagger\h{n^*,1} \vert^2}
    {\mathbf{v}_{n^*}^\dagger (\mathbf{I}_{N_r} +
    \mathbf{R}_{n^*})
     \mathbf{v}_{n^*}}\right)\NN.
\end{align}

To increase $\R$, various user selection schemes can be considered,
but we focus on several popular techniques in the following
subsections -- to maximize the postprocessed SNR
$\big(\textrm{i.e.,~} P\vert \mathbf{v}_{n^*}^\dagger\h{n^*,1}
\vert^2  \big)$,
to minimize the postprocessed INR $\big(\textrm{i.e.,~}
\mathbf{v}_{n^*}^\dagger \mathbf{R}_{n^*} \mathbf{v}_{n^*} \big)$,
and to maximize the postprocessed SINR $\Big(\textrm{i.e.,~} \frac{
P\vert \mathbf{v}_{n^*}^\dagger\h{n^*,1} \vert^2}
{\mathbf{v}_{n^*}^\dagger (\mathbf{I}_{N_r} + \mathbf{R}_{n^*})
\mathbf{v}_{n^*}} \Big)$.

%%%%%%%%%%%%%%%%%%%%%%%%%%%%%%%%%%%%%%%%%%%%%%%%%%%%%%%%%%%%%%%%%%%
% % % % % % % % % % % % % % % % % % % % % % % % % % % % % % % % % %
%%%%%%%%%%%%%%%%%%%%%%%%%%%%%%%%%%%%%%%%%%%%%%%%%%%%%%%%%%%%%%%%%%%
\subsection{The Maximum Postprocessed SNR User Selection (MAX-SNR)}

In the MAX-SNR user selection scheme, each user maximizes the
postprocessed SNR, and the transmitter selects the user having the
maximum postprocessed SNR. Consequently, the postprocessed SNR at
the selected user becomes
\begin{align}
 \max_n \left[ \max_{\mathbf{v}_n}
 ~ P\vert \mathbf{v}_{n}^\dagger\h{n,1}\vert^2\right]
 \stackrel{(a)}{=} \max_n P\Vert \h{n,1}\Vert^2,
\end{align}
where the equality $(a)$ holds when the $n$th user adopts the
postprocessing vector $\mathbf{v}_n^\SNR = \mathbf{h}_{n,1} /
\Vert\mathbf{h}_{n,1}\Vert$.
Thus, the selected user denoted by $n^*_\SNR$ becomes
\begin{align}
    n_\SNR^* = \argmax{n}~ P\Vert \h{n,1}\Vert^2,
\end{align}
and the desired channel gain at each user ($\Vert \h{n,1}\Vert^2$
for the user $n$) should be informed to the transmitter.

%%%%%%%%%%%%%%%%%%%%%%%%%%%%%%%%%%%%%%%%%%%%%%%%%%%%%%%%%%%%%%%%%
Using the MAX-SNR scheme, the transmitter can only increase the DoF
gain term while the DoF loss term remains one. Although the MAX-SNR
scheme can achieve the target DoF $d$ by the strategy
$(d_1,d_2)=(1+d,1)$, it is not optimal target DoF achieving
strategy. The required user scaling for the target DoF $d~(> 0)$ by
the MAX-SNR scheme becomes
\begin{align}
 N\propto e^{P^{(1+d)}},\NN
\end{align}
as shown in Theorem \ref{theorem:N_scaling_DoF_gain}. This user
scaling is of course higher than \eqref{eqn:scaling(d<1)} for the
target DoF $d \in (0,1]$ and \eqref{eqn:scaling(d>1)} for the target
DoF $d ~(> 1)$ since the MAX-SNR does not realize the optimal target
achieving strategy. In other words, one can easily find that
\begin{align}
 &\lim_{P\to\infty} \big[e^{P^{(1+d)}}/ P^{d(K-N_r)}\big]
 =\infty \qquad~~~~~~~\textrm{for~~}d\in(0,1],\NNL
 &\lim_{P\to\infty} \big[e^{P^{(1+d)}}/ \big(e^{P^{(d-1)}} P^{(K-N_r)}\big)\big]
 =\infty \quad\textrm{for~~}d>1.\NN
\end{align}

%%%%%%%%%%%%%%%%%%%%%%%%%%%%%%%%%%%%%%%%%%%%%%%%%%%%%%%%%%%%%%%%%%%
% % % % % % % % % % % % % % % % % % % % % % % % % % % % % % % % % %
%%%%%%%%%%%%%%%%%%%%%%%%%%%%%%%%%%%%%%%%%%%%%%%%%%%%%%%%%%%%%%%%%%%
\subsection{The Minimum Postprocessed INR User Selection (MIN-INR)}

In the MIN-INR user selection scheme, each user minimizes the
postprocessed INR, and the transmitter selects the user having the
minimum postprocessed INR.
Thus, the postprocessed INR at the selected user becomes
\begin{align}
 \min_n \left[ \min_{\mathbf{v}_n}
    ~ \mathbf{v}_{n}^\dagger \mathbf{R}_{n} \mathbf{v}_{n}
 \right]
 \stackrel{(a)}{=}
    \min_n \left[\Lambda _{\min} \left(\mathbf{R}_{n}\right)\right],
\end{align}
where the equality $(a)$ is obtained by the postprocessing vector of
the $n$th user
\begin{align}
 \mathbf{v}_n^\INR = V_{\min} \left(\mathbf{R}_{n}\right). \label{eqn:v_n_INR}
\end{align}
The required feedback information from the $n$th user is $\Lambda
_{\min} \big(\mathbf{R}_n\big)$, and index of the selected user
denoted by $n_\INR^*$ becomes
\begin{align}
 n_\INR^* = \argmin{n} \left[\Lambda _{\min}
 \left(\mathbf{R}_{n}\right)\right].
\end{align}
Note that this scheme minimizes the rate loss term defined in
\eqref{eqn:Rloss}.

%%%%%%%%%%%%%%%%%%%%%%%%%%%%%%%%%%%%%%%%%%%%%%%%%%%%%%%%%%%%%%%%%
Using the MIN-INR scheme, the transmitter can decrease the DoF loss
term while the DoF gain term remains to be one. Therefore, the
MIN-INR scheme realizes the optimal target DoF achieving strategy
$(d_1,d_2)=(1,1-d)$ for the target DoF $d\in[0,1]$. The required
number of users by the MIN-INR scheme for the target DoF $d\in[0,1]$
scales like
\begin{align}
 N\propto P^{d(K-N_r)},\NN
\end{align} which is the required user scaling of the optimal
target DoF achieving strategy when the target DoF is $d\in [0,1]$ as
shown in Theorem \ref{theorem:scaling1}.

%%%%%%%%%%%%%%%%%%%%%%%%%%%%%%%%%%%%%%%%%%%%%%%%%%%%%%%%%%%%%%%%%

\JHnewpage

\JHnewpage

%%%%%%%%%%%%%%%%%%%%%%%%%%%%%%%%%%%%%%%%%%%%%%%%%%%%%%%%%%%%%%%%%%%
% % % % % % % % % % % % % % % % % % % % % % % % % % % % % % % % % %
%%%%%%%%%%%%%%%%%%%%%%%%%%%%%%%%%%%%%%%%%%%%%%%%%%%%%%%%%%%%%%%%%%%
\subsection{The Maximum Postprocessed SINR User Selection (MAX-SINR)}
The MAX-SINR user selection scheme is known to maximize the
achievable rate at the transmitter although it requires additional
complexity for postprocessing at the receivers. The achievable rate
by the MAX-SINR scheme denoted by $\R_\SINR$ becomes
\begin{align}
 \R_\SINR \triangleq \mathbb{E}\left[\max_{n,\mathbf{v}_n}~
    \log_2
    \left(1+
    \frac{ P\vert \mathbf{v}_n^\dagger\h{n,1} \vert^2}
    {\mathbf{v}_n^\dagger (\mathbf{I}_{N_r} +
    \mathbf{R}_n) \mathbf{v}_n}
    \right)\right]. \label{eqn:R_SINR}
\end{align}
At each channel realization, the postprocessed SINR at the selected
user is given by
\begin{align}
 \max_n\left[\max_{\mathbf{v}_n}~
    \frac{ P\vert \mathbf{v}_n^\dagger\h{n,1} \vert^2}
    {\mathbf{v}_n^\dagger (\mathbf{I}_{N_r} +
    \mathbf{R}_n)
    \mathbf{v}_n}
    \right]. \label{eqn:MAX_SINR}
\end{align}
To maximize the postprocessed SINR, the $n$th user adopts the
postprocessing vector given by
\begin{align}
 \mathbf{v}_n^\SINR
 = \frac{(\mathbf{I}_{N_r} + \mathbf{R}_n)^{-1}\mathbf{h}_{n,1}}
    {\Vert(\mathbf{I}_{N_r} +
    \mathbf{R}_n)^{-1}\mathbf{h}_{n,1}\Vert}, \NN
\end{align}
which is identical with the MMSE-IRC in \cite{OMAAT2011}.
The corresponding postprocessed SINR at user $n$ becomes $P
\mathbf{h}_{n,1}^\dagger (\mathbf{I}_{N_r} + \mathbf{R}_n)^{-1}
\mathbf{h}_{n,1}$ \cite{TRWJ2004}, and hence the selected user at
the transmitter denoted by $n_\SINR^*$ is given by
\begin{align}
 n_\SINR^* = \argmax{n} ~
 \mathbf{h}_{n,1}^\dagger (\mathbf{I}_{N_r} + \mathbf{R}_n)^{-1}
 \mathbf{h}_{n,1} .
\end{align}

%%%%%%%%%%%%%%%%%%%%%%%%%%%%%%%%%%%%%%%%%%%%%%%%%%%%%%%%%%%%%%%%%%%
% % % % % % % % % % % % % % % % % % % % % % % % % % % % % % % % % %
%%%%%%%%%%%%%%%%%%%%%%%%%%%%%%%%%%%%%%%%%%%%%%%%%%%%%%%%%%%%%%%%%%%
\begin{lemma}\label{lemma:scaling1}
To obtain the target DoF $d\in[0,1]$, the required user scaling of
the MAX-SINR scheme is exactly the same as that of the MIN-INR
scheme.
\end{lemma}
%%%%%%%%%%%%%%%%%%%%%%%%%%%%%%%%%%%%%%%%%%%%%%%%%%%%%%%%%%%%%%%%%
\begin{proof}
From the fact that
\begin{align}
 \R_\INR \le \R_\SINR
 \le \R^+_\SINR - \R^-_\INR,
\end{align}
we obtain
\begin{align}
 &\DoF{\R_\INR} \le \DoF{\R_\SINR}
 \le \DoF{\R^+_\SINR} - \DoF{\R^-_\INR}
    \stackrel{(a)}{=} \DoF{\R_\INR},
\end{align}
where the equality $(a)$ is because $\tDoF{\R^+_\SINR}=1$ as shown
in the proof of Theorem \ref{theorem:scaling1}. Therefore, the
required user scaling for $\tDoF{\R_\SINR}=d$ is exactly the same as
the required user scaling for $\tDoF{\R_\INR}=d$, equivalently, for
$\tDoF{\R^-_\INR}=1-d$.
\end{proof}

Lemma \ref{lemma:scaling1} indicates that the MAX-SINR scheme
realizes the optimal DoF achieving strategy $(d_1^\star,
d_2^\star)=(1,1-d)$ for the target DoF $d\in[0,1]$.

%%%%%%%%%%%%%%%%%%%%%%%%%%%%%%%%%%%%%%%%%%%%%%%%%%%%%%%%%%%%%%%%%%%
% % % % % % % % % % % % % % % % % % % % % % % % % % % % % % % % % %
%%%%%%%%%%%%%%%%%%%%%%%%%%%%%%%%%%%%%%%%%%%%%%%%%%%%%%%%%%%%%%%%%%%
\begin{lemma}\label{lemma:scaling2}
The MAX-SINR scheme realizes the DoF achieving strategy $(d_1^\star,
d_2^\star) = (d, 0)$ whenever the target DoF $d$ is greater than 1.
\end{lemma}
%%%%%%%%%%%%%%%%%%%%%%%%%%%%%%%%%%%%%%%%%%%%%%%%%%%%%%%%%%%%%%%%%
\begin{proof}
Since the MAX-SINR scheme is the optimal user selection scheme, it
achieves DoF $d~(>1)$ with the user scaling $N \propto e^{P^{(d-1)}}
P^{(K-N_r)}$ as stated in Theorem \ref{theorem:scaling2}.
From the definition of \eqref{eqn:Rgain_max}, we obtain $\R_\SINR^+
< \eqref{eqn:Rgain_max}$, and hence we have $\tDoF{\R_\SINR^+} \le
\tDoF{\eqref{eqn:Rgain_max}}$.
As shown in \eqref{eqn:still_d}, the sufficient user scaling for the
MAX-SINR scheme to obtain the target DoF $d$ cannot increase the DoF
gain term larger than $d$ even if the whole user scaling is only
devoted to increasing the DoF gain term.
This implicates that when we obtain the target DoF $d~(>1)$ by the
MAX-SINR scheme with the user scaling $N \propto e^{P^{(d-1)}}
P^{(K-N_r)}$, we obtain $\tDoF{\R_\SINR^-} = 0$ and have the DoF
gain $d$ at most (i.e., $\tDoF{\R_\SINR^+} = d$). Therefore, the
MAX-SINR scheme can only have $\left(\tDoF{\R_\SINR^+},
\tDoF{\R_\SINR^-}\right)=(d,0)$ if $\tDoF{\R_\SINR}=d~(>1)$.
\end{proof}

%%%%%%%%%%%%%%%%%%%%%%%%%%%%%%%%%%%%%%%%%%%%%%%%%%%%%%%%%%%%%%%%%%%
% % % % % % % % % % % % % % % % % % % % % % % % % % % % % % % % % %
%%%%%%%%%%%%%%%%%%%%%%%%%%%%%%%%%%%%%%%%%%%%%%%%%%%%%%%%%%%%%%%%%%%
\subsection{Two-stage User Selection Scheme}

For the target DoF $d~(>1)$, the two-stage user selection scheme
described in the proof of Lemma \ref{lemma:general_scaling_d>1} can
be adopted. More specifically, the transmitter selects the users by
the MAX-SNR scheme in the first stage. Then, in the second stage,
the transmitter selects a single user by the MIN-INR scheme or the
MAX-SINR scheme. As shown in the proof of Lemma
\ref{lemma:general_scaling_d>1}, the two-stage user selection scheme
can realize the optimal target DoF achieving strategy for the target
DoF $d~(>1)$.

%\newpage
%%%%%%%%%%%%%%%%%%%%%%%%%%%%%%%%%%%%%%%%%%%%%%%%%%%%%%%%%%%%%%%%%%%
% % % % % % % % % % % % % % % % % % % % % % % % % % % % % % % % % %
%%%%%%%%%%%%%%%%%%%%%%%%%%%%%%%%%%%%%%%%%%%%%%%%%%%%%%%%%%%%%%%%%%%
\section{Extension to $K$-transmitter Interfering MIMO Broadcast Channels}
\label{sec:extension}

In this section, we extend our system to interfering MIMO BC cases.
More specifically, each transmitter with $N_t$ antennas sends $N_t$
independent streams over $N_t$ orthonormal random beams using equal
power allocation.
Similar to the user selection procedure in \cite{SH2005}, each
transmitter broadcasts $N_t$ orthonormal random beams, and each user
feeds $N_t$ scalar values corresponding to all beams back to the
transmitter. The feedback information corresponding to each stream
such as SNR, INR, and SINR can be easily found in a similar way to
the SIMO case.
A single user is selected for each beam, but the same user can be
selected for different beams. However, it rarely occurs that
multiple streams are transmitted for a single user as the number of
users increases. When multiple streams are transmitted for a single
user, the user is assumed to decode each stream treating the other
streams as interferences. We denote the orthonormal random beams by
$\mathbf{u}_1, \ldots, \mathbf{u}_{N_t}$ which satisfies that
$\Vert\mathbf{u}_1 \Vert^2 = \cdots =\Vert \mathbf{u}_{N_t} \Vert^2
= 1$ and $\mathbf{u}_i^\dagger \mathbf{u}_j=0$ for all $i\ne j$. We
start from the following remark.

%%%%%%%%%%%%%%%%%%%%%%%%%%%%%%%%%%%%%%%%%%%%%%%%%%%%%%%%%%%%%%%%%%%
% % % % % % % % % % % % % % % % % % % % % % % % % % % % % % % % % %
%%%%%%%%%%%%%%%%%%%%%%%%%%%%%%%%%%%%%%%%%%%%%%%%%%%%%%%%%%%%%%%%%%%
\begin{remark}Let $\mathbf{H}\in\mathbb{C}^{N_r\times N_t}$
be the channel matrix whose elements are i.i.d. Gaussian random
variables.
Then, for an arbitrary unitary matrix $\mathbf{U}\in
\mathbb{C}^{N_t\times N_t}$ (i.e., $\mathbf{U} ^\dagger \mathbf{U} =
\mathbf{U} \mathbf{U}^\dagger = \mathbf{I}_{N_t}$), the
distributions of $\mathbf{H}$ and $\mathbf{HU}$ are identical.
Since the $N_r$ columns of $\mathbf{H}$ are independent and
isotropic random vectors in $\mathbb{C}^{N_r}$, so are the $N_r$
columns of $\mathbf{HU} \in \mathbb{C}^{N_r\times N_t}$.
\end{remark}
%%%%%%%%%%%%%%%%%%%%%%%%%%%%%%%%%%%%%%%%%%%%%%%%%%%%%%%%%%%%%%%%%%%
% % % % % % % % % % % % % % % % % % % % % % % % % % % % % % % % % %
%%%%%%%%%%%%%%%%%%%%%%%%%%%%%%%%%%%%%%%%%%%%%%%%%%%%%%%%%%%%%%%%%%%

Owing to Remark 1, the $K$-transmitter MIMO interfering BC is
statistically identical with the $KN_t$-transmitter SIMO interfering
BC.
Let $\mathbf{H}_{k,n}\in\mathbb{C}^{N_r\times N_t}$ be the channel
matrix from the $k$th transmitter to the $n$th user in the first
user group. Since the random beams satisfy that $[\mathbf{u}_1,
\ldots, \mathbf{u}_{N_t}]^\dagger [\mathbf{u}_1, \ldots,
\mathbf{u}_{N_t}] = [\mathbf{u}_1, \ldots, \mathbf{u}_{N_t}]
[\mathbf{u}_1, \ldots, \mathbf{u}_{N_t}]^\dagger =
\mathbf{I}_{N_t}$, the user $n$ in the first user group has $KN_t$
independent and isotropic channel vectors
\begin{align}
 \mathbf{H}_{k,n}\mathbf{u}_i\in\mathbb{C}^{N_r\times1}\qquad
    k\in[K] ~~ i\in[N_t],\NN
\end{align}
formed by the random beams and channel matrices from all
transmitters.

If the $n$th user in the first group is served by the $i$th random
beam, the user has desired channel $\mathbf{H}_{n,1}\mathbf{u}_i
\in\mathbb{C}^{N_t\times 1}$ and the $(KN_t-1)$ interfering
channels, which correspond to $(N_t-1)$ inter-stream interfering
channels $\{\mathbf{H}_{n,1} \mathbf{u}_j\}_{j\ne i}$ and $(K-1)N_t$
inter-transmitter interfering channels
\begin{align}
 \bigcup_{j\in[N_t]}
    \{\mathbf{H}_{n,k}\mathbf{u}_j\}_{k=2}^K.\NN
\end{align}
Consequently, each random beam can be regarded as a single antenna
transmitter with the transmit power $P/N_t$. This fact leads to the
following theorems.

%%%%%%%%%%%%%%%%%%%%%%%%%%%%%%%%%%%%%%%%%%%%%%%%%%%%%%%%%%%%%%%%%%%
% % % % % % % % % % % % % % % % % % % % % % % % % % % % % % % % % %
%%%%%%%%%%%%%%%%%%%%%%%%%%%%%%%%%%%%%%%%%%%%%%%%%%%%%%%%%%%%%%%%%%%
\begin{theorem}[MIMO BC]\label{theorem:MIMO_BC} In a MIMO BC
where a transmitter with $N_t$ antennas supports $N_t$ users among
$N$ users with $N_r (< N_t)$ antennas each, the optimal DoF
achieving strategy for the target DoF $d~(\in [0,N_t])$ is $(N_t,
N_t-d)$ and requires the number of users to scale as $N\propto
P^{(d/N_t)(N_t-N_r)}$. For the target DoF $d~(>N_t)$, the optimal
DoF achieving strategy is $(d, 0)$ and requires the number of users
to scale as $N \propto e^{P^{(d/N_t-1)}} P^{(N_t-N_r)}$.
\end{theorem}
%%%%%%%%%%%%%%%%%%%%%%%%%%%%%%%%%%%%%%%%%%%%%%%%%%%%%%%%%%%%%%%%%%%
\begin{proof} The DoF gain term $d_1 (>N_t)$ is obtained
when each stream achieves DoF gain term $d_1/N_t ~(>1)$. Thus, with
the same procedure given in Section \ref{sec:DoF_gain_scaling}, we
can easily show that the DoF gain term $d_1 (>N_t)$ is obtained when
$N\propto e^{P^{(d_1/N_t-1)}}$.
On the other hand, the DoF loss term $d_2~(\le N_t)$ is obtained
when each stream achieves DoF loss term per stream $d_2/N_t~ (\le
1)$. As stated earlier, using $N_t$ orthonormal random beams at the
transmitter each of which independently supports a single user, the
MIMO BC can be translated into an $N_t$-transmitter SIMO IC where
each transmitter supports one of $N$ users with $N_r$ antennas.
In this case, DoF loss term $d_2~(\le N_t)$, i.e., DoF loss
$d_2/N_t~(\le 1)$ per stream, is obtained when $N\propto
P^{(1-d_2/N_t)(N_t-N_r)}$.
Therefore, we can conclude that the optimal DoF achieving strategy
for the target DoF $d~(\in [0,N_t])$ is $(N_t, N_t-d)$ and requires
the number of users to scale as $N\propto P^{(d/N_t)(N_t-N_r)}$.
Also, for the target DoF $d~(>N_t)$, the optimal DoF achieving
strategy is $(d, 0)$ and requires the number of users to scale as $N
\propto e^{P^{(d/N_t-1)}} P^{(N_t-N_r)}$.
\end{proof}
%%%%%%%%%%%%%%%%%%%%%%%%%%%%%%%%%%%%%%%%%%%%%%%%%%%%%%%%%%%%%%%%%%%
% % % % % % % % % % % % % % % % % % % % % % % % % % % % % % % % % %
%%%%%%%%%%%%%%%%%%%%%%%%%%%%%%%%%%%%%%%%%%%%%%%%%%%%%%%%%%%%%%%%%%%

%%%%%%%%%%%%%%%%%%%%%%%%%%%%%%%%%%%%%%%%%%%%%%%%%%%%%%%%%%%%%%%%%%%
% % % % % % % % % % % % % % % % % % % % % % % % % % % % % % % % % %
%%%%%%%%%%%%%%%%%%%%%%%%%%%%%%%%%%%%%%%%%%%%%%%%%%%%%%%%%%%%%%%%%%%
\begin{theorem}[Interfering MIMO BC]\label{theorem:MIMO_IBC}
Consider a $K$-transmitter interfering MIMO BC where the $k$th
transmitter with $N_t^{(k)}$ antennas supports $N_t^{(k)}$ users
among $N^{(k)}$ users with $N_r^{(k)} (<T\triangleq \sum_{k=1}^K
N_t^{(k)})$ antennas each.
At the $k$th transmitter, the optimal DoF achieving strategy for the
target DoF $d~(\in [0,N_t^{(k)}])$ is $(N_t^{(k)}, N_t^{(k)}-d)$ and
requires the number of users to scale as $N\propto
P^{(d/N_t^{(k)})(T-N_r^{(k)})}$.
For the target DoF $d ~(>N_t)$, the optimal DoF achieving strategy
becomes $(d, 0)$ and requires the number of users to scale as
$N^{(k)} \propto e^{P^{(d/N_t^{(k)}-1)}} P^{(T-N_r^{(k)})}$.
\end{theorem}
%%%%%%%%%%%%%%%%%%%%%%%%%%%%%%%%%%%%%%%%%%%%%%%%%%%%%%%%%%%%%%%%%%%
\begin{proof}
The proof is similar to that of Theorem \ref{theorem:MIMO_BC}.
The $k$th transmitter obtains DoF gain term $d_1~(>N_t^{(k)})$ when
each stream obtains DoF gain term $d_1/N_t^{(k)} (>1)$, and the
required user scaling is exactly given by $N^{(k)} \propto
e^{P^{(d_1/N_t^{(k)}-1)}}$.
On the other hand, the $k$th transmitter obtains $d_2~(\le
N_t^{(k)})$ when the DoF loss term per stream becomes
$d_2/N_t^{(k)}~ (\le 1)$.
Using $N_t^{(k)}$ orthonormal random beams at each transmitter each
of which independently supports a single user, the interfering MIMO
BC can be translated into an $T(=\sum_{k=1}^K
N_t^{(k)})$-transmitter SIMO IC where each transmitter supports a
single user among $N^{(k)}$ users with $N_r^{(k)}$ antennas.
Thus, the $k$th transmitter obtains the DoF loss term $d_2~(\le
N_t^{(k)})$ when $N^{(k)} \propto P ^{(1- d_2/N_t^{(k)})
(T-N_r^{(k)})}$.
Therefore, we can conclude that the optimal DoF achieving strategy
of the $k$th transmitter for the target DoF $d~(\in [0,N_t^{(k)}])$
is $(N_t^{(k)}, N_t^{(k)}-d)$ and obtained when the number of users
scales as $N^{(k)}\propto P^{(d/N_t^{(k)})(T-N_r^{(k)})}$.
Also, for the target DoF $d~(>N_t^{(k)})$, the optimal DoF achieving
strategy is $(d, 0)$ and requires the number of users to scale as
$N^{(k)} \propto e^{P^{(d/N_t^{(k)}-1)}} P^{(T-N_r^{(k)})}$.
\end{proof}

\color{black}

\JHnewpage
%%%%%%%%%%%%%%%%%%%%%%%%%%%%%%%%%%%%%%%%%%%%%%%%%%%%%%%%%%%%%%%%%%%
% % % % % % % % % % % % % % % % % % % % % % % % % % % % % % % % % %
%%%%%%%%%%%%%%%%%%%%%%%%%%%%%%%%%%%%%%%%%%%%%%%%%%%%%%%%%%%%%%%%%%%
\section{Numerical Results}

%%%\begin{figure}[!t]
%%%\centering
%%%  \includegraphics[width=.9777\columnwidth]{./figures/OIA_K_4_a10_N10_new2.eps}\\
%%%  \caption{Achievable rates per transmitter using various schemes in
%%%    IBC with $(K, N_r)=(4, 3)$ and $N=10$.}
%%%  \label{fig:OIA_K4_a10_N10}
%%%\end{figure}
%%%
%%%\begin{figure}[!t]
%%%\centering
%%%  \includegraphics[width=.9777\columnwidth]{./figures/OIA_K_4_N_P_new2.eps}\\
%%%  \caption{Achievable rates per transmitter using various schemes
%%%  when  the number of users in each group scales as
%%%  $N=P$ and $N=0.5P$, respectively.}
%%%  \label{fig:OIA_K4_a10_kN}
%%%\end{figure}
%%%
%%%\begin{figure}[!t]
%%%\centering
%%%  \includegraphics[width=.9777\columnwidth]{./figures/OIA_K_4_N_Ppower_new2.eps}\\
%%%  \caption{Achievable rates per transmitter using various schemes
%%%  when the number of users in each group scales as
%%%  $N=P^{0.5}$ and $N=P^{1}$, respectively.}
%%%  \label{fig:OIA_K4_a10_N_powerP}
%%%\end{figure}

In this section, we first compare achievable rates of the practical
user selection schemes for given number of users. Then, we check if
the target DoF can be achievable with increasing number of users by
showing achievable rates per transmitter for the practical user
selection schemes.
We have also considered two time division multiple access (TDMA)
schemes.
In the first TDMA scheme (TDMA1), a single transmitter operates at
each time so that $1/K$ DoF is achieved at each transmitter.
In the second TDMA scheme (TDMA2), only $N_r$ of $K$ transmitters
operate at each time so that $N_r/K$ DoF is achieved at each
transmitter.

Fig. \ref{fig:OIA_K4_a10_N10} shows the achievable rates of each
transmitter for various user selection schemes in IBC when there are
4 transmitters and each transmitter has 10 users with three receive
antennas each. It is confirmed that the achievable rates are
saturated in the high SNR region and the achievable DoF per
transmitter becomes zero for the fixed number of users.

Now, we show that the target DoF can be achievable if the number of
users properly scales. In Fig. \ref{fig:OIA_K4_a10_kN}, the number
of users scales as $N \propto P^{d(K-N_r)}$, i.e., $N\propto P$ for
the target DoF one. Specifically, two user scaling $N=P$ and
$N=0.5P$ are considered, and other configurations except the number
of users are the same as those in Fig. \ref{fig:OIA_K4_a10_N10}.
Fig. \ref{fig:OIA_K4_a10_kN} verifies that the MIN-INR and the
MAX-SINR schemes achieve DoF one per transmitter as predicted in
Theorem \ref{theorem:scaling1} and Lemma \ref{lemma:scaling1}.

In Fig. \ref{fig:OIA_K4_a10_N_powerP}, we consider two different
user scaling $N=P^{0.5}$ and $N=P^{1}$ from those in Fig.
\ref{fig:OIA_K4_a10_kN}.
According to Theorem \ref{theorem:scaling1} and Lemma
\ref{lemma:scaling1}, the achievable DoF at each transmitter by
either the MAX-SINR scheme or the MIN-INR scheme is $d$ when the
number of users scales as $N \propto P^d$. As predicted, Fig.
\ref{fig:OIA_K4_a10_N_powerP} shows that the achieved DoF per
transmitter is 0.5 and 1 when $N=P^{0.5}$ and $N=P^{1}$,
respectively, by either the MIN-INR scheme or the MAX-SINR scheme.

%%%%%%%%%%%%%%%%%%%%%%%%%%%%%%%%%%%%%%%%%%%%%%%%%%%%%%%%%%%%%%%%%%%
% % % % % % % % % % % % % % % % % % % % % % % % % % % % % % % % % %
%%%%%%%%%%%%%%%%%%%%%%%%%%%%%%%%%%%%%%%%%%%%%%%%%%%%%%%%%%%%%%%%%%%
%\JHnewpage
\section{Conclusions}

We first studied the optimal way of exploiting multiuser diversity
in the $K$-transmitter SIMO IBC where each transmitter with a single
antenna selects a user and the number of transmitters is larger than
the number of receive antennas at each user. We proved that the
multiuser dimensions should be used first for decreasing the DoF
loss caused by interfering signals; the whole multiuser dimensions
should be exploited to reduce the DoF loss term to $1-d$ for the
target DoF $d\in[0, N_t]$, while the multiuser dimensions should be
devoted to making the DoF loss zero and then to increasing the DoF
gain term to $d$ for the target DoF $d \in [N_t,\infty)$. We also
derived the sufficient user scaling for the target DoF. The DoF per
transmitter $d\in[0,N_t]$ is obtained when the number of users
scales as $N \propto P^{(d/N_t)(KN_t-N_r)}$, and the DoF per
transmitter $d \in [N_t,\infty)$ is achieved when the number of
users scales as $N \propto e^{P^{(d/N_t-1)}} P^{(KN_t-N_r)}$.
Also, we extended the results to the $K$-transmitter MIMO IBC where
each transmitter having the multiple antennas supports the multiple
users.

%%%%%%%%%%%%%%%%%%%%%%%%%%%%%%%%%%%%%%%%%%%%%%%%%%%%%%%%%%%%%%%%%%%
% % % % % % % % % % % % % % % % % % % % % % % % % % % % % % % % % %
%%%%%%%%%%%%%%%%%%%%%%%%%%%%%%%%%%%%%%%%%%%%%%%%%%%%%%%%%%%%%%%%%%%
% % % % % % % % % % % % % % % % % % % % % % % % % % % % % % % % % %
%%%%%%%%%%%%%%%%%%%%%%%%%%%%%%%%%%%%%%%%%%%%%%%%%%%%%%%%%%%%%%%%%%%
% % % % % % % % % % % % % % % % % % % % % % % % % % % % % % % % % %
%%%%%%%%%%%%%%%%%%%%%%%%%%%%%%%%%%%%%%%%%%%%%%%%%%%%%%%%%%%%%%%%%%%
% % % % % % % % % % % % % % % % % % % % % % % % % % % % % % % % % %
%%%%%%%%%%%%%%%%%%%%%%%%%%%%%%%%%%%%%%%%%%%%%%%%%%%%%%%%%%%%%%%%%%%

%\linespread{1.777}

\newpage

\begin{figure}[!t]
\centering
  \includegraphics[width=.5777\columnwidth]{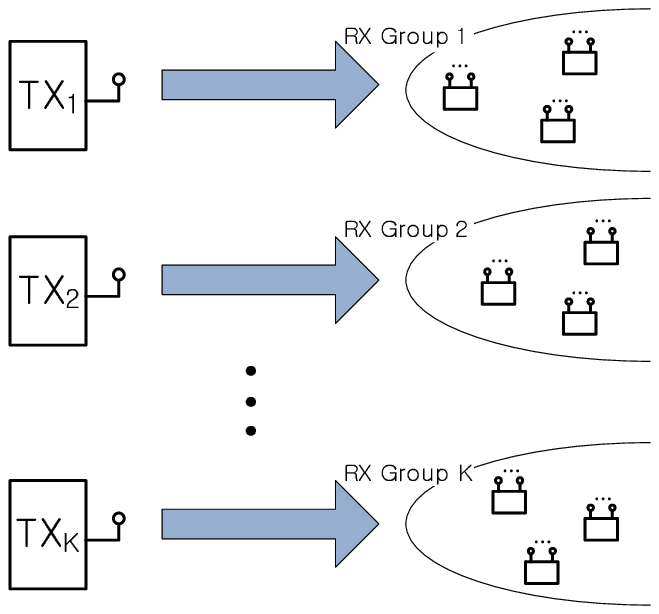}\\
  \caption{System model. Each transmitter selects and serves a single user in each group.}
  \label{fig:system_model}
\end{figure}

\begin{figure}[!t]
\centering
  \includegraphics[width=.8777\columnwidth]{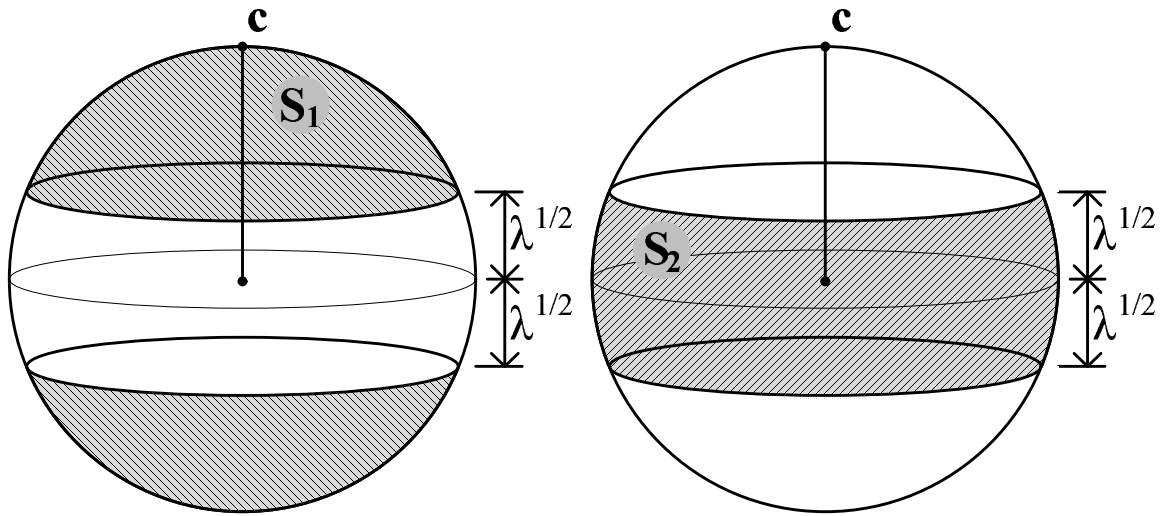}
  \caption{Illustration of $S_1(\mathbf{c}, \lambda)$ and $S_2(\mathbf{c}, \lambda)$
  in $\mathbb{R}^3$ case.}
  \label{fig:sphere_S1S2}
\end{figure}

\begin{figure}[!t]
\centering
  \includegraphics[width=\columnwidth]{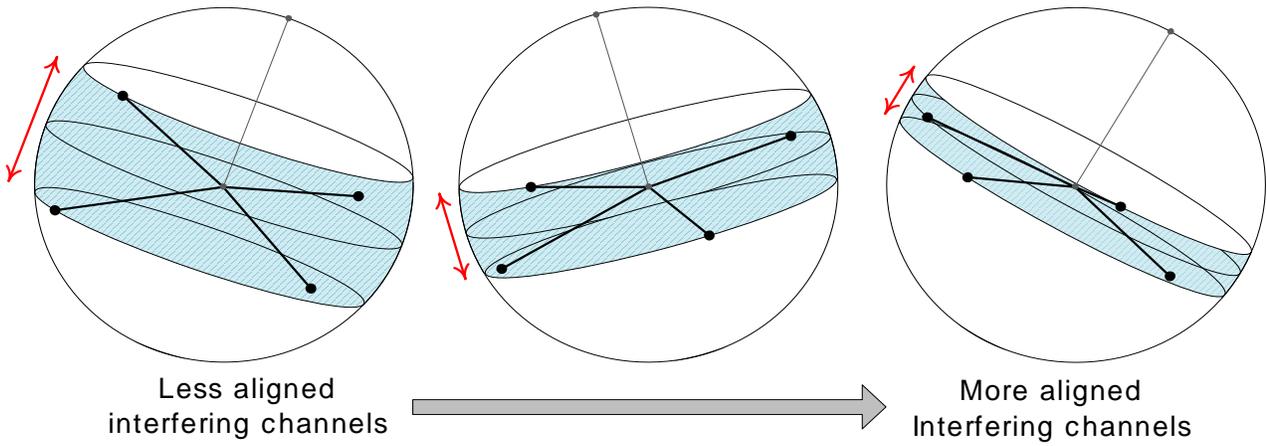}
  \caption{Graphical representations of the interfering channels
  and interference alignment measures.}
  \label{fig:IAM_concept}
\end{figure}

\begin{figure}[!t]
\centering
  \includegraphics[width=.8777\columnwidth]{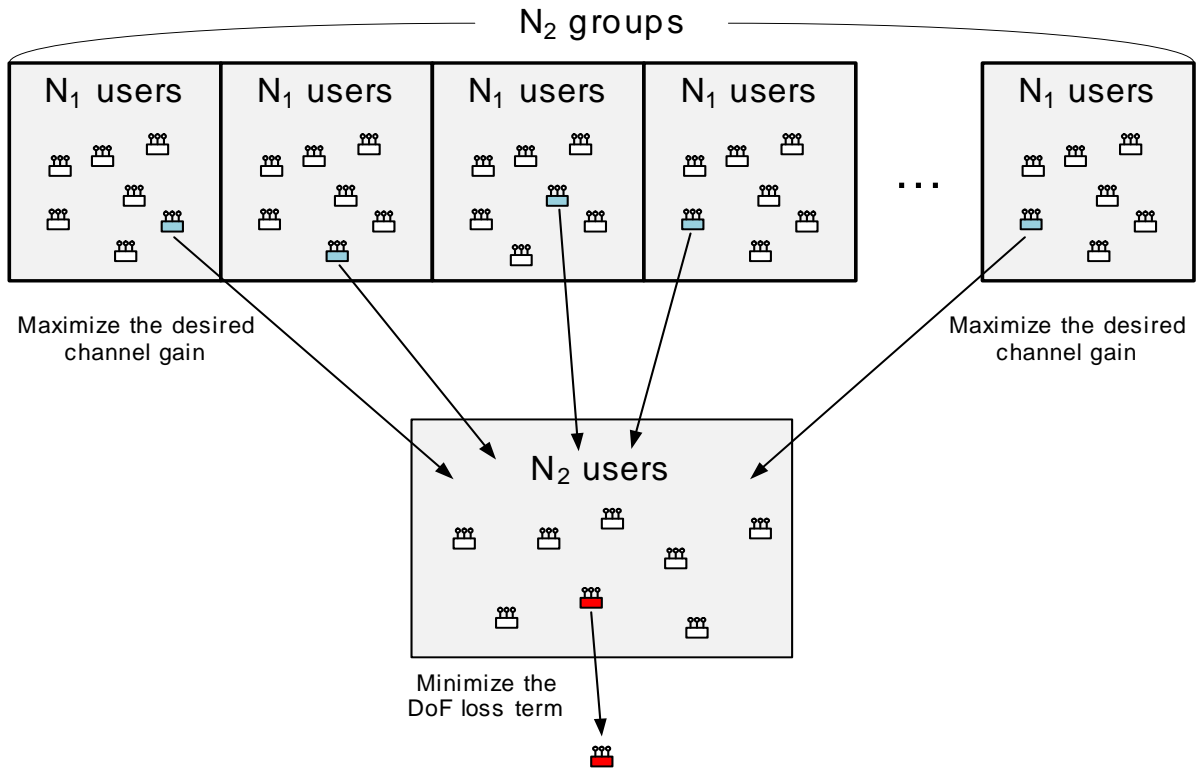}
  \caption{Two-stage user selection scheme.}
  \label{fig:user_selection_2stage}
\end{figure}

\begin{figure}[!t]
\centering
  \includegraphics[width=.6777\columnwidth]{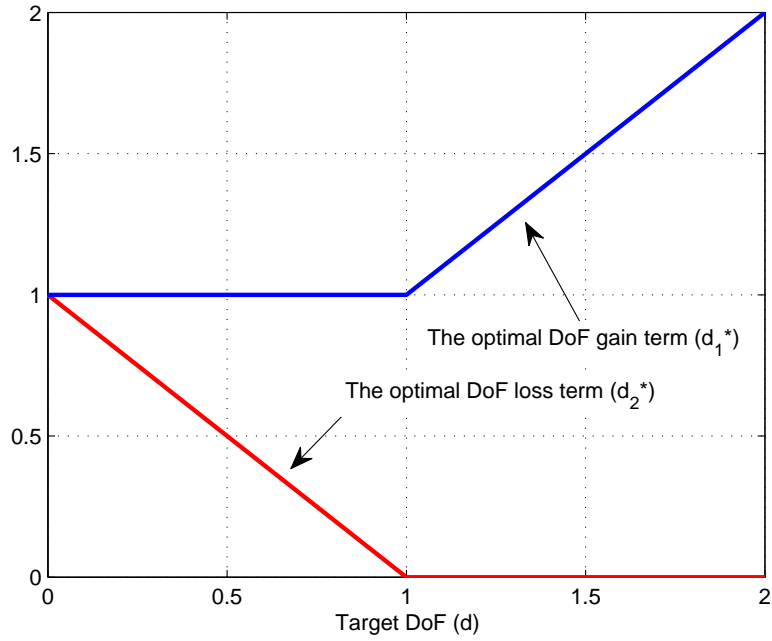}
  \caption{The optimal DoF achieving strategy $(d_1^\star, d_2^\star )$
    for the target DoF $d$}
  \label{fig:DoF_achieving_strategy}
\end{figure}

\begin{figure}[!t]
\centering
  \includegraphics[width=.6777\columnwidth]{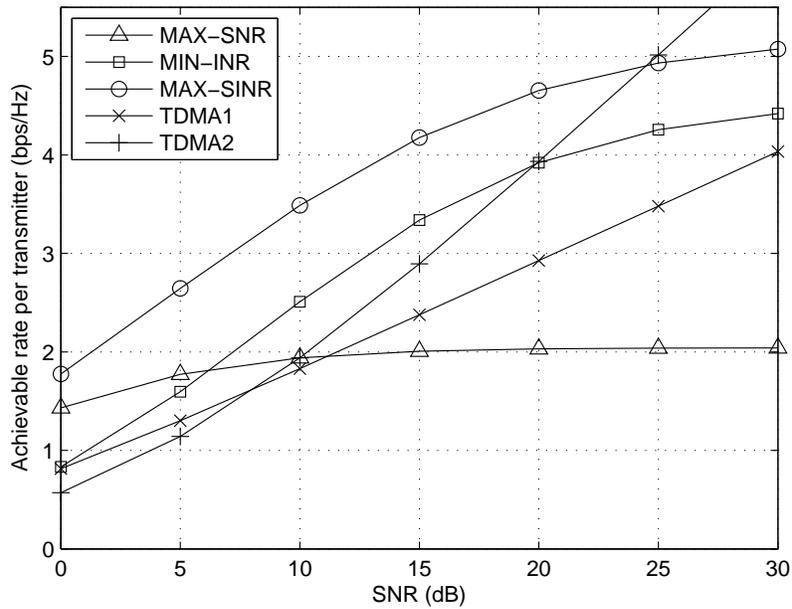}\\
  \caption{Achievable rates per transmitter using various schemes in
    IBC with $(K, N_r)=(4, 3)$ and $N=10$.}
  \label{fig:OIA_K4_a10_N10}
\end{figure}

\begin{figure}[!t]
\centering
  \includegraphics[width=.6777\columnwidth]{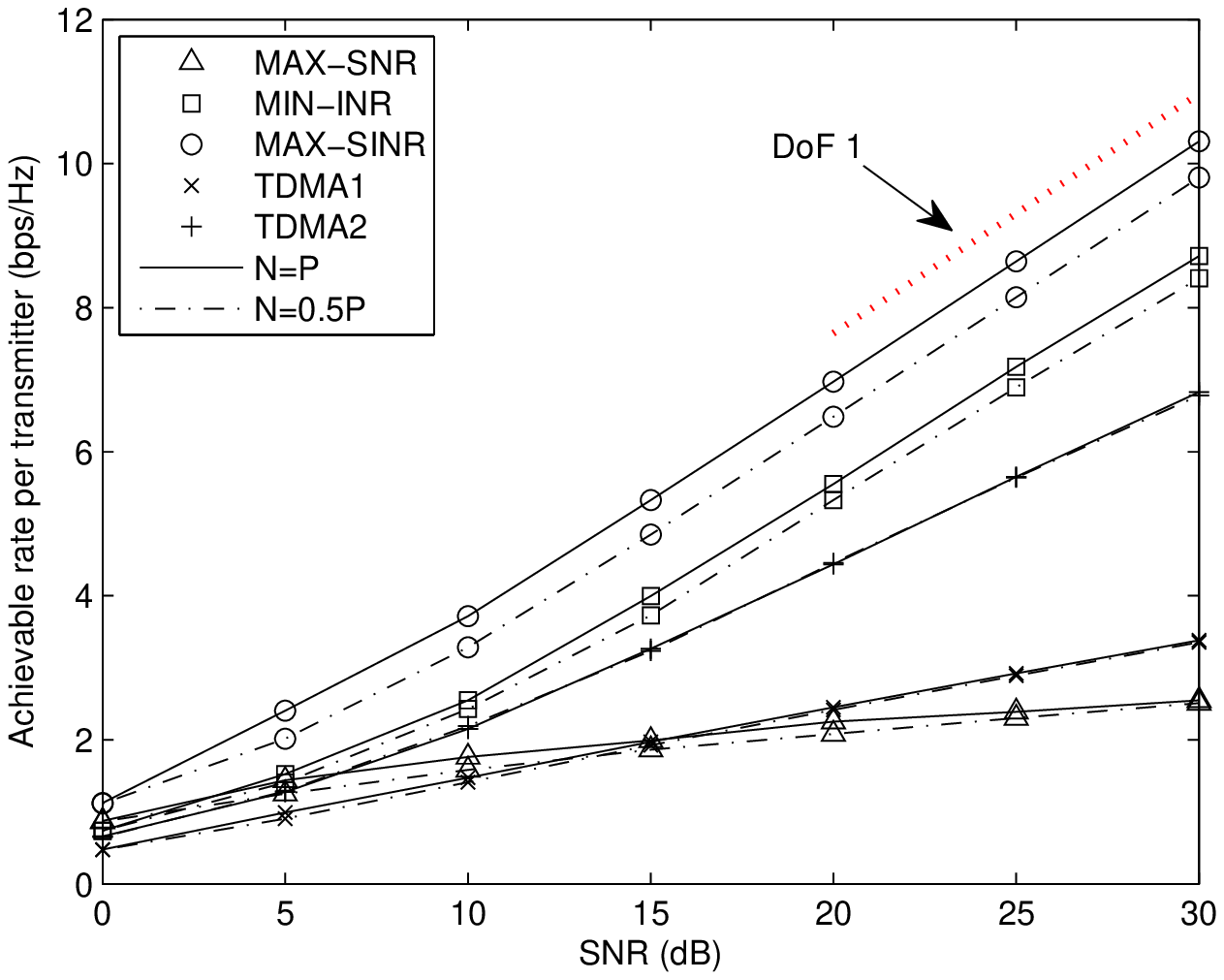}\\
  \caption{Achievable rates per transmitter using various schemes
  when  the number of users in each group scales as
  $N=P$ and $N=0.5P$, respectively.}
  \label{fig:OIA_K4_a10_kN}
\end{figure}

\begin{figure}[!t]
\centering
  \includegraphics[width=.6777\columnwidth]{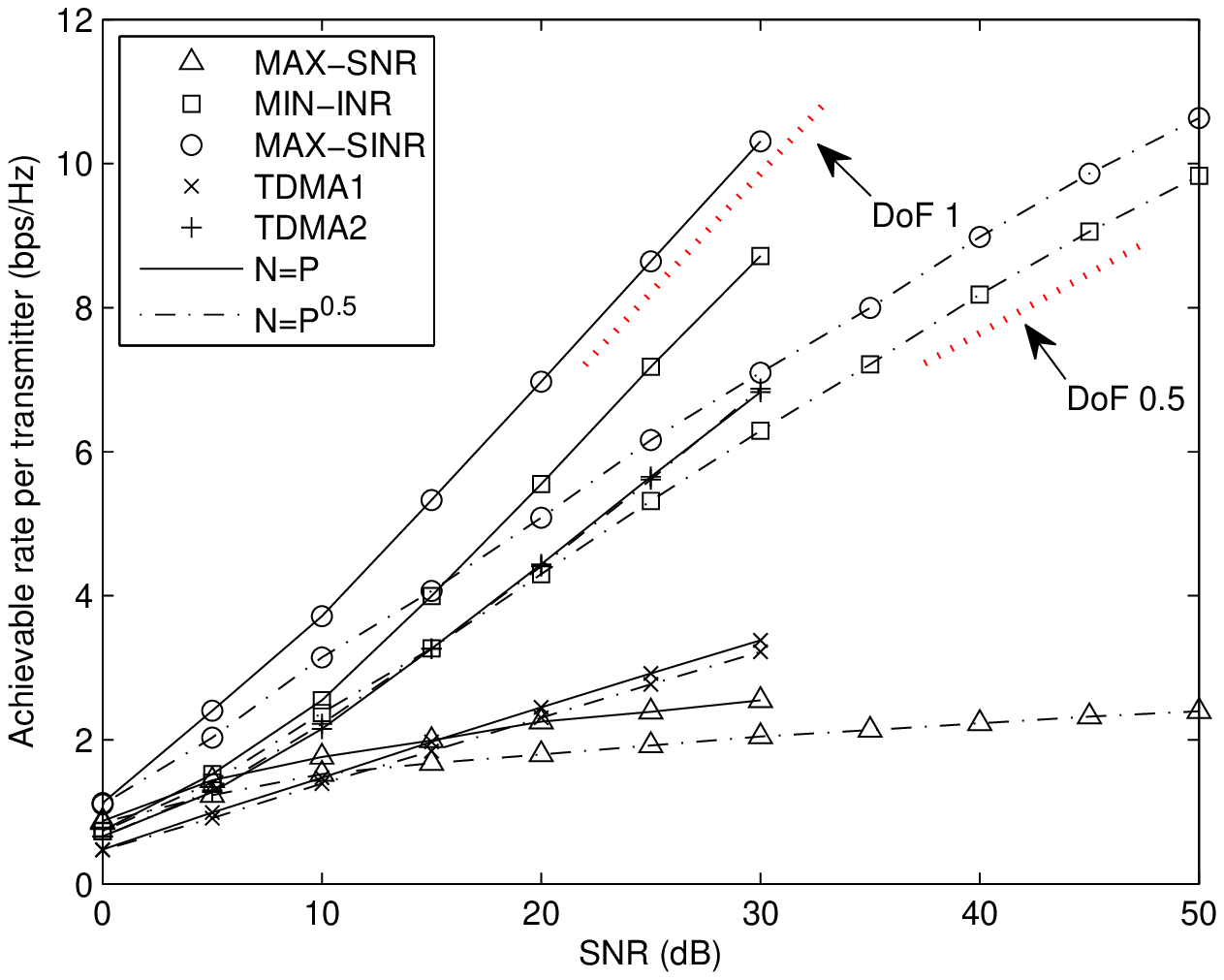}\\
  \caption{Achievable rates per transmitter using various schemes
  when the number of users in each group scales as
  $N=P^{0.5}$ and $N=P^{1}$, respectively.}
  \label{fig:OIA_K4_a10_N_powerP}
\end{figure}

\end{document}